%% file: main.tex
\newcommand{\longversion}[1]{#1}
\newcommand{\shortversion}[1]{}

\documentclass[a4paper,UKenglish,cleveref, autoref, thm-restate]{lipics-v2021}

\title{Knapsack with Vertex Cover, Set Cover, and Hitting Set} 

 \author{Palash Dey}{Indian Institute of Technology Kharagpur, India \and \url{https://cse.iitkgp.ac.in/~palash/} }{palash.dey@cse.iitkgp.ac.in}{ https://orcid.org/0000-0003-0071-9464}{}

\author{Ashlesha Hota}{Indian Institute of Technology Kharagpur, India}{ashleshahota.23@kgpian.iitkgp.ac.in}{https://orcid.org/0009-0009-8805-4583}{}

\author{Sudeshna Kolay}{Indian Institute of Technology Kharagpur, India \and \url{https://cse.iitkgp.ac.in/~skolay/}}{skolay@cse.iitkgp.ac.in}{https://orcid.org/0000-0002-2975-4856}{}

\author{Sipra Singh}{Indian Institute of Technology Kharagpur, India }{sipra.singh@iitkgp.ac.in}{}{}

\authorrunning{Anonymous authors} 



\category{} 

\relatedversion{} 

\nolinenumbers 

\EventEditors{}
\EventNoEds{2}
\EventLongTitle{}
\EventShortTitle{}
\EventAcronym{}
\EventYear{}
\EventDate{}
\EventLocation{}
\EventLogo{}
\SeriesVolume{}
\ArticleNo{}

\input{preamble.tex}
\begin{document}

\maketitle


\input{abstract}
\input{introduction}
\input{prelim}
\input{results-classical}
\input{results-poly-approx}

\input{results-pc}
\input{conclusion}

\newpage
\bibliography{references}

\appendix
\

\end{document}

%% file: preamble.tex
\usepackage{longtable}
\newcommand{\tabitem}{~~\llap{\textbullet}~~}

\usepackage{hyperref}
\usepackage{amssymb,amsmath}
\usepackage{color}
\usepackage{makecell}
\usepackage{mathtools}
\usepackage{bbm}
\usepackage{setspace}

\usepackage[usenames,svgnames,dvipsnames]{xcolor}

\usepackage{tikz}
\usetikzlibrary{arrows,positioning}

\newdimen\prevdp
\def\leftlabel#1{\noalign{\prevdp=\prevdepth
   \kern-\prevdp\nointerlineskip\vbox to0pt{\vss\hbox{\ensuremath{#1}}}\kern\prevdp}}

\usepackage{url}

\newcommand\numberthis{\addtocounter{equation}{1}\tag{\theequation}}

\usepackage{exscale ,relsize}

\usepackage{enumerate}

\usepackage{xspace}
\newcommand{\NP}{\ensuremath{\mathsf{NP}}\xspace}
\newcommand{\NPC}{\ensuremath{\mathsf{NP}}\text{-complete}\xspace}
\newcommand{\NPH}{\ensuremath{\mathsf{NP}}\text{-hard}\xspace}

\newcommand{\el}{\ensuremath{\ell}\xspace}

\newcommand{\FPT}{\ensuremath{\mathsf{FPT}}\xspace}

\newcommand{\OPT}{\ensuremath{\mathsf{OPT}}\xspace}
\newcommand{\ALG}{\ensuremath{\mathsf{ALG}}\xspace}

\newcommand{\Pb}{\ensuremath{\mathsf{P}}\xspace}
\newcommand{\tw}{\text{tw}\xspace}

\let\oldlambda\lambda
\renewcommand{\lambda}{\ensuremath{\oldlambda}\xspace}
\let\oldalpha\alpha
\renewcommand{\alpha}{\ensuremath{\oldalpha}\xspace}
\let\oldDelta\Delta
\renewcommand{\Delta}{\ensuremath{\oldDelta}\xspace}

\newcommand{\yes}{{\sc yes}\xspace}

\newcommand{\vcknapsack}{{\sc Vertex Cover Knapsack}\xspace}
\newcommand{\vcknapsackbudget}{{\sc $k$-Vertex Cover Knapsack}\xspace}
\newcommand{\minimumvcknapsack}{{\sc Minimum Vertex Cover Knapsack}\xspace}
\newcommand{\minimalvcknapsack}{{\sc Minimal Vertex Cover Knapsack}\xspace}

\newcommand{\kp}{\text{\sc Knapsack}\xspace}
\newcommand{\vc}{\text{\sc Vertex Cover}\xspace}
\newcommand{\maxminimalvc}{{\sc Maximum Minimal Vertex Cover}\xspace}

\newcommand{\setc}{\text{\sc Set Cover}\xspace}
\newcommand{\setck}{\text{\sc Set Cover Knapsack}\xspace}

\newcommand{\hs}{\text{d-\sc Hitting Set}\xspace}
\newcommand{\hsk}{\text{d-\sc Hitting Set Knapsack}\xspace}

\newcommand{\twd}{{\text{\sc Treewidth}}\xspace}

\renewcommand{\AA}{\ensuremath{\mathcal A}\xspace}
\newcommand{\BB}{\ensuremath{\mathcal B}\xspace}
\newcommand{\CC}{\ensuremath{\mathcal C}\xspace}

\newcommand{\EE}{\ensuremath{\mathcal E}\xspace}
\newcommand{\FF}{\ensuremath{\mathcal F}\xspace}
\newcommand{\GG}{\ensuremath{\mathcal G}\xspace}

\newcommand{\II}{\ensuremath{\mathcal I}\xspace}
\newcommand{\JJ}{\ensuremath{\mathcal J}\xspace}

\newcommand{\NN}{\ensuremath{\mathcal N}\xspace}
\newcommand{\OO}{\ensuremath{\mathcal O}\xspace}

\renewcommand{\SS}{\ensuremath{\mathcal S}\xspace}
\newcommand{\TT}{\ensuremath{\mathcal T}\xspace}
\newcommand{\UU}{\ensuremath{\mathcal U}\xspace}
\newcommand{\VV}{\ensuremath{\mathcal V}\xspace}
\newcommand{\WW}{\ensuremath{\mathcal W}\xspace}
\newcommand{\XX}{\ensuremath{\mathcal X}\xspace}

\newcommand{\NB}{\ensuremath{\mathbb N}\xspace}

\usepackage{nicefrac}

\newcommand{\eps}{\ensuremath{\varepsilon}\xspace}
\renewcommand{\epsilon}{\eps}

\newcommand{\ignore}[1]{}

\newcommand{\pr}{\ensuremath{\prime}}

\renewcommand{\leq}{\leqslant}
\renewcommand{\geq}{\geqslant}
\renewcommand{\ge}{\geqslant}
\renewcommand{\le}{\leqslant}

\usepackage{enumitem}
\setlist[enumerate]{labelwidth=!, labelindent=0pt,topsep=0pt}

\usepackage{rotating}
\usepackage{amssymb}
\usepackage{latexsym}
\usepackage{epsfig}
\usepackage{xcolor}
\usepackage{url}
\usepackage{fancyhdr}

\usepackage{caption}
\usepackage{subcaption}
\usepackage{array}
\usepackage{multirow}
\usepackage{enumerate}
\usepackage{pdflscape}
\usepackage{amsmath,graphicx,algorithm,algpseudocode,amsfonts}
\usepackage{epstopdf}
\usepackage{float}
\floatstyle{ruled}
\newfloat{Algorithms}{!thb}{lop}
\floatname{Algorithms}{Algorithm}

\allowdisplaybreaks[1]

\algnewcommand\algorithmicinput{\textbf{Input:}}
\algnewcommand\INPUT{\item[\algorithmicinput]}
\algnewcommand\algorithmicoutput{\textbf{Output:}}
\algnewcommand\OUTPUT{\item[\algorithmicoutput]}
\algnewcommand{\LineComment}[1]{\State \(\triangleright\) #1}

\usepackage{pifont}

\usepackage{mathtools}

\usepackage{algorithm}
\usepackage{algpseudocode}


\usepackage{comment}

\usepackage{cleveref}

\crefname{theorem}{Theorem}{\bf Theorems}
\crefname{observation}{Observation}{\bf Observations}
\crefname{corollary}{Corollary}{\bf Corollary}
\crefname{lemma}{Lemma}{\bf Lemmata}
\crefname{corollary}{Corollary}{\bf Corollaries}
\crefname{proposition}{Proposition}{\bf Propositions}
\crefname{definition}{Definition}{\bf Definitions}
\crefname{claim}{Claim}{\bf Claims}
\crefname{reductionrule}{Reduction rule}{\bf Reduction rules}
\usepackage{color,soul} 

%% file: abstract.tex
\begin{abstract}
In the \vcknapsack problem, we are given an undirected graph $\GG=(\VV,\EE)$, with weights $(w(u))_{u\in\VV}$ and values $(\alpha(u))_{u\in\VV}$ of the vertices, the size $s$ of the knapsack, a target value $p$, and the goal is to compute if there exists a vertex cover $\UU\subseteq\VV$ with total weight at most $s$, and total value at least $p$. This problem simultaneously generalizes the classical vertex cover and knapsack problems. We show that this problem is strongly \NPC. However, it admits a pseudo-polynomial time algorithm for trees. In fact, we show that there is an algorithm that runs in time $\OO\left(2^{\tw}\cdot n\cdot {\sf min}\{s^2,p^2\}\right)$ where $\tw$ is the treewidth of \GG. Moreover, we can compute a $(1-\eps)$- approximate solution for maximizing the value of the solution given the knapsack size as input in time $\OO\left(2^{\tw}\cdot poly(n,1/ \epsilon,\log\left(\sum_{v\in\VV}\alpha(v)\right))\right)$ and a $(1+\eps)$-approximate solution to minimize the size of the solution given a target value as input, in time $\OO\left(2^{\tw}\cdot poly(n,1/ \epsilon,\log\left(\sum_{v\in\VV}w(v)\right))\right)$ for every $\eps>0$. Restricting our attention to polynomial-time algorithms only, we then consider polynomial-time algorithms and present a $2$ factor polynomial-time approximation algorithm for this problem for minimizing the total weight of the solution, which is optimal up to additive $o(1)$ assuming Unique Games Conjecture (UGC). On the other hand, we show that there is no $\rho$ factor polynomial-time approximation algorithm for maximizing the total value of the solution given a knapsack size for any $\rho>1$ unless $\Pb=\NP$.

Furthermore, we show similar results for the variants of the above problem when the solution \UU needs to be a minimal vertex cover, minimum vertex cover, and vertex cover of size at most $k$ for some input integer $k$. Then, we consider set families (equivalently hypergraphs) and study the variants of the above problem when the solution needs to be a set cover and hitting set. We show that there are $H_d$ and $f$ factor polynomial-time approximation algorithms for \setck where $d$ is the maximum cardinality of any set and $f$ is the maximum number of sets in the family where any element can belongs in the input for minimizing the weight of the knapsack given a target value, and a $d$ factor polynomial-time approximation algorithm for \hsk which are optimal up to additive $o(1)$ assuming UGC. On the other hand, we show that there is no $\rho$ factor polynomial-time approximation algorithm for maximizing the total value of the solution given a knapsack size for any $\rho>1$ unless $\Pb=\NP$ for both \setck and \hsk.
\end{abstract}

%% file: introduction.tex
\section{Introduction}
\label{sec:introduction}

A {\em vertex cover} of an undirected graph is a set of vertices that contains at least one endpoint of every edge. For a real-world application of vertex cover, consider a city network \GG where the vertices are the major localities of the city, and we have an edge between two vertices if the distance between their corresponding locations is less than, say, five kilometers. A retail chain wants to open a few stores in the city in such a way that everyone (including the people living between any two major localities) in the city has a retail shop within five kilometers. The cost of opening a store depends on location. We can see that the company needs to compute a minimum weight vertex cover of the \GG to open stores with the minimum total cost, where the weight of a vertex is the cost of opening a store at that location. However, each store has the potential to generate non-core revenue, say from advertising. In such a scenario, the company may like to maximize the total non-core revenue without compromising its core business, which it will accomplish by opening stores at the vertices of a vertex cover. This is precisely what we call \vcknapsack. In this problem, we are given an undirected graph $\GG=(\VV,\EE)$, with weights $(w(u))_{u\in\VV}$ and values $(\alpha(u))_{u\in\VV}$ of the vertices, the size $s$ of the knapsack, a target value $p$, and the goal is to compute if there exists a vertex cover $\UU\subseteq\VV$ with $w(\UU)=\sum_{u\in\UU} w(u) \le s$, and $\alpha(\UU)=\sum_{u\in\UU} \alpha(u) \ge p$. 

We study several natural variations of this problem: (i) \vcknapsackbudget where the solution should be a vertex cover of size at most $k$ for an integer input $k$, (ii) \minimalvcknapsack where the solution should be a minimal vertex cover, and (iii) \minimumvcknapsack where the solution should be a minimum vertex cover. 

We then consider the hypergraphs or equivalently set families. There, we consider the knapsack generalization of the set cover and hitting set problems. In \setck, we are given a collection $S_1,\ldots,S_m$ of subsets of a universe $[n]$, with weights $(w(j))_{j\in[m]}$ and values $(\alpha(j))_{j\in[m]}$ for the sets, the size $s$ of the knapsack, a target value $p$, and the goal is to compute if there exists a set cover of total weight at most $s$ and total value at least $p$. On the other hand, we have a collection $S_1,\ldots,S_m$ of $d$ sized subsets of a universe $[n]$ with weights  in \hsk $(w(j))_{j\in[n]}$ and values $(\alpha(j))_{j\in[n]}$ for the elements, the size $s$ of the knapsack, a target value $p$, and the goal is to compute if there exists a hitting set of total weight at most $s$ and total value at least $p$.


\subsection{Contributions}

We study these problems under the lens of classical complexity theory, parameterized complexity, polynomial-time approximation, and FPT-approximation. We summarize our results in \Cref{tab:contribution}.


 We now give a high-level overview of the techniques used in our results. For the $f$-approximation algorithm for \setck, the dual LP of a configuration LP relaxation has two types of constraints: intuitively speaking, one set of constraints handles the knapsack part while the other set takes care of the set cover requirement. We first increase some dual variables iteratively so that some of the dual constraints corresponding to the knapsack part of the problem become tight. We pick the sets corresponding to these constraints. If this gives a valid set cover, then we are done. Otherwise, we increase some dual constraints iteratively corresponding to the set cover part of the problem until we satisfy the set cover requirements. The first part of our $H_d$-approximation algorithm is the same as the $f$-approximation algorithm. In the second part, we use the greedy algorithm for the set cover problem to pick more sets if the sets picked in the first part do not form a set cover.
 
\begin{longtable}{c|c}\hline
        Knapsack variant & Results\\\hline\hline
		Vertex Cover & \makecell[l]{\tabitem Strongly \NPC (\Cref{thm:vckp-gen-npc})\\\tabitem \NPC for star graphs (\Cref{vck-trees-npc})\\\tabitem Poly-time 2-approx. to minimize weight$^\dagger$ (\Cref{2factorvc})\\\tabitem Poly-time $\rho$-inapprox. to maximize value$^\star$ (\Cref{hardnessofapprox})\\\tabitem $\OO\left(2^{\tw}\cdot n^{\OO(1)}\cdot {\sf min}\{s^2,p^2\}\right)$ (\Cref{vckparameterzied})\\\tabitem $\OO\left(2^{\tw}\cdot poly(n,1/ \epsilon,\log\left(\sum_{v\in\VV}\alpha(v)\right))\right)$ time, $(1 - \epsilon)$\\ approximation to maximize value$^\star$ (\Cref{vc-ptas})\\\tabitem $\OO\left(2^{\tw}\cdot poly(n,1/ \epsilon,\log\left(\sum_{v\in\VV}w(v)\right))\right)$ time, $(1 + \epsilon)$\\ approximation to minimize weight$^\ddagger$ (\Cref{fptas-weight}) } \\\hline
        
		Vertex Cover of size $\le k$ & \makecell[l]{\tabitem Strongly \NPC (\Cref{thm:vckb-gen-npc})\\\tabitem\NPC for star graphs (\Cref{k-vck-trees-npc})\\\tabitem$\OO\left(2^{\tw}\cdot n^{\OO(1)}\cdot {\sf min}\{s^2,p^2\}\right)$ (\Cref{vckbudparameterzied})\\\tabitem$\OO\left(2^{\tw}\cdot poly(n, 1/ \epsilon,\log\left(\sum_{v\in\VV}\alpha(v)\right))\right)$ time, $(1 - \epsilon)$\\ approximation to maximize value$^\star$ (\Cref{vc-ptas})\\\tabitem $\OO\left(2^{\tw}\cdot poly(n,1/ \epsilon,\log\left(\sum_{v\in\VV}w(v)\right))\right)$ time, $(1 + \epsilon)$\\ approximation to minimize weight$^\ddagger$ (\Cref{fptas-weight})}\\\hline
  
		Minimum Vertex Cover & \makecell[l]{\tabitem\NP hard (\Cref{minimumvck-nph}) \\\tabitem\NPC for trees (\Cref{minimumvck-trees-npc})\\\tabitem$\OO\left(2^{\tw}\cdot n^{\OO(1)}\cdot {\sf min}\{s^2,p^2\}\right)$ (\Cref{minvckparameterzied}) \\\tabitem$\OO\left(2^{\tw}\cdot poly(n, 1/ \epsilon,,\log\left(\sum_{v\in\VV}\alpha(v)\right))\right)$ time, $(1 - \epsilon)$\\ approximation to maximize value$^\star$ (\Cref{vc-ptas})\\\tabitem $\OO\left(2^{\tw}\cdot poly(n,1/ \epsilon,\log\left(\sum_{v\in\VV}w(v)\right))\right)$ time, $(1 + \epsilon)$\\ approximation to minimize weight$^\ddagger$ (\Cref{fptas-weight})}\\\hline
  
		Minimal Vertex Cover & \makecell[l]{\tabitem Strongly \NPC (\Cref{thm:minimalvckp-gen-npc}) \\\tabitem\NPC for trees (\Cref{minimalvck-trees-npc})\\\tabitem No poly-time approx. algorithm neither to\\ maximize value$^\star$ nor to minimize weight$^\dagger$ (\Cref{hardnessofapprox})\\\tabitem $\OO\left(16^{\tw}\cdot n^{\OO(1)}\cdot {\sf min}\{s^2,p^2\}\right)$ (\Cref{fptminimal}) \\\tabitem$\OO\left(16^{\tw}\cdot poly(n, 1/ \epsilon,\log\left(\sum_{v\in\VV}\alpha(v)\right))\right)$ time, $(1 - \epsilon)$\\ approximation to maximize value$^\star$ (\Cref{thm-fptas})\\\tabitem $\OO\left(16^{\tw}\cdot poly(n,1/ \epsilon,\log\left(\sum_{v\in\VV}w(v)\right))\right)$ time, $(1 + \epsilon)$\\ approximation to minimize weight$^\ddagger$ (\Cref{fptas-weight})}\\\hline

         Set Cover & \makecell[l]{\tabitem Strongly \NPC (\Cref{setcover-knapsack-npc})\\\tabitem Poly-time f-approx. to minimize weight$^\dagger$ (\Cref{thm:approxf})\\\tabitem Poly-time $H_d$-approx. to minimize weight$^\dagger$ (\Cref{hgapprosetcover})\\\tabitem Poly-time $\rho$-inapprox. to maximize value$^\star$ (\Cref{hardnessofapprox})} \\ \hline

        $d$-Hitting Set & \makecell[l]{\tabitem Strongly \NPC (\Cref{hitting-knapsack-npc})\\\tabitem Poly-time d-approx. to minimize weight$^\dagger$ (\Cref{hsdfactor})\\\tabitem Poly-time $\rho$-inapprox. to maximize value$^\star$  (\Cref{hardnessofapprox})} \\ \hline \multicolumn{2}{c}{}\\

	\caption{Summary of results. $\tw:$ treewidth of the graph, $s:$ size of knapsack, $p:$ target value of knapsack, $\eps:$ any real number greater than zero, $n:$ number of vertices or size of the universe, $f:$ the maximum number of sets where any element belongs, $d:$ maximum size of any set, $\rho:$ any poly-time computable function. $\star:$ size of knapsack is input. $\dagger:$ bag size is input. $\ddagger:$ target value is input.}
	\label{tab:contribution}
\end{longtable}

Our fixed-parameter pseudo-polynomial time algorithms with respect to treewidth for the variants of vertex cover knapsack combine the idea of pseudo-polynomial time algorithm and the dynamic programming algorithm for vertex cover with respect to treewidth. Then, we use these algorithms in a black-box fashion to obtain \FPT-approximation algorithms.

\longversion{ 
\subsection{Organization}
The rest of the paper is organized as follows. 
We formally define the problems in Section \ref{pd}. We show the classical NP-Completeness results in Section \ref{npcresults}. We design polynomial-time approximation algorithms and \FPT algorithm in Section \ref{approximation} and Section \ref{parameterized}, respectively. We conclude our efforts in Section \ref{conclusion}. A preliminary
version of this work was published before \cite{dey2024vcknapsack} which did not contain most of the proofs.
}

\subsection{Related Work}
\label{related}

The classical knapsack problem admits a fully polynomial time approximation scheme (FPTAS)~\cite{vazirani2001approximation,williamson2011design}. Many extensions and generalizations of the knapsack problem have been studied extensively~\cite{salkin1975knapsack,cacchiani2022knapsack,pisinger2007quadratic}. Since our paper focuses on generalizations of knapsack to some graph theoretic problems and their extension to hypergraphs, we discuss only those knapsack variants directly related to ours.

Yamada et al.~\cite{yamada2002heuristic} proposed heuristics for knapsack when there is a graph on the items, and the solutions need to be an independent set. Many intractability results in special graph classes and heuristic algorithms based on pruning, dynamic programming, etc. have been studied for this independent set knapsack problem~\cite{hifi2006reactive,hifi2007reduction,PferschyS09,BettinelliCM17,HeldCS12,ConiglioFS21,LuizSU21,PferschyS17,GurskiR19,GoebbelsGK22,DBLP:journals/dam/BonomoE19,DBLP:journals/orl/ManninoORC07}. Dey et al.~\cite{dey2024knapsack} studied the knapsack problem with graph-theoretic constraints like - connectedness, paths, and shortest path. 

Note that our \vcknapsack also generalizes the classical weighted vertex cover problem, for which we know a polynomial-time $2$-approximation algorithm which is the best possible approximation factor up to additive $\eps>0$ that one can achieve in polynomial time assuming Unique Games Conjecture~\cite{vazirani2001approximation,williamson2011design}. On the parameterized side, there is a long line of work on designing a fast \FPT algorithm for vertex cover parameterized by the size $k$ of a minimum vertex cover, with the current best being $\OO\left(1.25284^k\cdot n^{\OO(1)}\right)$~\cite{DBLP:conf/stacs/0001N24}. A related problem is maximum minimal vertex cover, where the goal is to compute a minimal vertex cover of maximum cardinality. This problem admits a polynomial kernel in terms of the number of vertices and can be solved in $\OO\left(2^k+n^{\OO(1)}\right)$ time~\cite{fernau2005parameterized}. Later, Peter Damaschke~\cite{damaschke2011parameterized} proved that it is solvable in time $\OO\left(1.62^k\cdot n^{\OO(1)}\right)$. Boria et al.~\cite{boria2015max} showed that there is a polynomial time $n^{-1/2}$ approximation algorithm and inapproximable within the ratio $n^{\epsilon-1/2}$ in polynomial time unless $\Pb=\NP$, where $\epsilon > 0$. 

Various approximation algorithms have been studied for the \setc problem with approximation ratios $f$ where $f$ is the maximum number of sets that any element can belong and $H_d$ where $d$ is the maximum cardinality of any set, and $H_d$ is the $d$-th harmonic number. These approximation factors are tight up to additive $\eps>0$ under standard complexity-theoretic assumptions~\cite{DBLP:journals/jacm/Feige98,DBLP:journals/jcss/KhotR08,williamson2011design,vazirani2001approximation}.



%% file: prelim.tex
\section{Preliminaries}
\label{pd}

We denote the set $\{1,2,\ldots\}$ of natural numbers with \NB. For any integer \el, we denote the sets $\{1,\ldots,\el\}$ and $\{0,1,\ldots,\el\}$ by $[\el]$ and $[\el]_0$ respectively. We now define our problems formally. Our first problem is \vcknapsack, where we need to find a vertex cover that fits the knapsack and meets a target value. A {\em vertex cover} of a graph is a subset of vertices that includes at least one end-point of every edge. Formally, it is defined as follows.

\begin{definition}[\vcknapsack]\label{def:vckp}
	Given an undirected graph $\GG=(\VV,\EE)$, weight of vertices $(w(u))_{u\in\VV}$, value of vertices $(\alpha(u))_{u\in\VV}$, size $s$ of knapsack, and target value $p$, compute if there exists a vertex cover $\UU\subseteq\VV$ of \GG with weight $w(\UU)=\sum_{u\in\UU} w(u) \le s$ and value $\alpha(\UU)=\sum_{u\in\UU} \alpha(u) \ge p$. We denote an any instance of it by $(\GG,(w(u))_{u\in\VV},(\alpha(u))_{u\in\VV},s,p)$.
\end{definition}

The \vcknapsackbudget, \minimumvcknapsack, \minimalvcknapsack problems are the same as \Cref{def:vckp} except we require the solution \UU to be respectively a vertex cover of size at most $k$ for an input integer $k$, a minimum cardinality vertex cover, a minimal vertex cover.

The treewidth of a graph quantifies the tree likeness of a graph~\cite{cygan2015parameterized}. Informally speaking, a tree decomposition of a graph is a tree where every node of the tree corresponds to some subsets of vertices, called bags, and the tree should satisfy three properties: (i) every vertex of the graph should belong to some bag, (ii) both the endpoints of every edge should belong to some bag, and (iii) the set of nodes of the tree containing any vertex should be connected..

The treewidth of a graph is defined as follows~\cite{cygan2015parameterized}.
\begin{definition}[\twd] 
Let $G = (V_G,E_G)$ be a graph.  A \textit{tree-decomposition} of a graph $G$ is a pair 
$(\mathbb{T} = ((V_{\mathbb{T}},E_{\mathbb{T}}),\mathcal{ X}=\{X_{t}\}_{t\in V_{\mathbb{ T}}})$,  where 
${\mathbb{ T}}$ is a tree where every node $t\in V_{\mathbb{ T}}$ 
is assigned a subset $X_t\subseteq V_G$, called a bag,  such that the following conditions hold. 
\begin{itemize}
\item $\bigcup_{t\in V_\mathbb{T}}{X_t}=V_G$,
\item for every edge $\{x,y\}\in E_G$ there is a $t\in V_\mathbb{T}$ such that  $x,y\in X_{t}$, and 
\item for any $v\in V_G$ the subgraph of $\mathbb{T}$ induced by the set  $\{t\mid v\in X_{t}\}$ is connected.
\end{itemize}
\end{definition}
The \textit{width} of a tree decomposition is $\max_{t\in V_\mathbb{T}} |X_t| -1$. The \textit{treewidth} of $G$ 
is the  minimum width over all tree decompositions of $G$ and is denoted by $\tw(G)$. 
 
A tree decomposition $(\mathbb{T},\mathcal{ X})$ is called a \textit{nice edge tree decomposition} if $\mathbb{T}$ is a tree rooted at some node $r$ where $ X_{r}=\emptyset$, each node of $\mathbb{T}$ has at most two children, and each node is of one of the following kinds:
\begin{itemize}
\item {\bf Introduce node}: a node $t$ that has only one child $t'$ where $X_{t}\supset X_{t'}$ and  $|X_{t}|=|X_{t'}|+1$.
\item {\bf  Forget vertex node}: a node $t$ that has only one child $t'$  where $X_{t}\subset X_{t'}$ and  $|X_{t}|=|X_{t'}|-1$.
\item {\bf Join node}:  a node  $t$ with two children $t_{1}$ and $t_{2}$ such that $X_{t}=X_{t_{1}}=X_{t_{2}}$.
\item {\bf Leaf node}: a node $t$ that is a leaf of $\mathbb T$, and $X_{t}=\emptyset$. 
\end{itemize}
We additionally require that every edge is introduced exactly once. 
One can  show that  a tree decomposition of width $t$ can be transformed into 
a nice tree decomposition of the same width $t$ and  with 
 $\mathcal{O}(t |V_G|)$ nodes \cite{cygan2015parameterized}. For a node $t \in \mathbb{T}$, let $\mathbb{T}_t$ be the subtree of $\mathbb{T}$ rooted at $t$, and $V(\mathbb{T}_t)$ denote the vertex set in that subtree. Then $G_t$ is the subgraph of $G$ where the vertex set is  $\bigcup_{t' \in V(\mathbb{T}_t)} X_{t'}$ and the edge set is the union of the set of edges introduced in each $t', t' \in V(\mathbb{T}_t)$. We denote by $V(G_t)$ the set of vertices in that subgraph, and by $E(G_E)$ the set of edges of the subgraph.: a node $t$ that is a leaf of $\mathbb T$, and $X_{t}=\emptyset$. 

 In this paper, we sometimes fix a vertex $v\in V_G$ and include it in every bag of a nice edge tree decomposition $(\mathbb{T},\mathcal{X})$ of $G$, with the effect of the root bag and each leaf bag containing $v$. For the sake of brevity, we also call such a modified tree decomposition a nice tree decomposition. Given the tree $\mathbb{T}$ rooted at the node $r$, for any nodes $t_1,t_2 \in V_\mathbb{T}$, the distance between the two nodes in $\mathbb{T}$ is denoted by $\sf{dist}_\mathbb{T}(t_1,t_2)$.

Extending the notion of vertex cover to hypergraphs, we define the \setck problem where we need to compute a set cover that fits the knapsack and achieves a maximum value. Formally, we define it as follows.

\begin{definition}[\setck]
    Given a collection $\FF=\{S_1,\ldots,S_m\}$ of subsets of a universe $[n]$ with weights $(w_j)_{j\in[m]}$ and values $(\alpha_j)_{j\in[m]}$ of the sets, size $s$ of knapsack, and target value $p$, compute if there exists a set cover $\JJ\subseteq[m]$ of weight $w(\JJ)=\sum_{j\in\JJ}w_j\le s$ and value $\alpha(\JJ)=\sum_{j\in\JJ}\alpha_j\ge p$. We denote any instance of it by $([n],\FF,(w_j)_{j\in[m]},(\alpha_j)_{j\in[m]},s,p)$.
\end{definition}

We also define \hsk, where we need to compute a hitting set that fits the knapsack and achieves at least the target value; here, items have weights and values.

\begin{definition}[\hsk]
    Given a collection $\FF=\{S_1,\ldots,S_m\}$ of subsets of a universe $[n]$ of size at most $d$ with weights $(w_i)_{i\in[n]}$ and values $(\alpha_i)_{i\in[n]}$ of the items, size $s$ of knapsack, and target value $p$, compute if there exists a hitting set $\II\subseteq[n]$ of weight $w(\II)=\sum_{i\in\II}w_i\le s$ and value $\alpha(\II)=\sum_{i\in\II}\alpha_i\ge p$. We denote any instance of it by $([n],\SS,(w_i)_{i\in[n]},(\alpha_i)_{i\in[n]},s,p)$.
\end{definition}

If not mentioned otherwise, we use $n$ to denote the number of vertices for problems involving graphs and the size of the universe for problems involving a set family; $m$ to indicate the number of edges for problems involving graphs and the number of sets in the family of sets for problems involving a set family; \tw to denote the treewidth of the graph; $s$ to represent the size of knapsack, and $p$ to denote the target value of solution.

%% file: results-classical.tex
\section{Results: Classical NP Completeness}
\label{npcresults}

In this section, we present our \NP-completeness results. Our first results show that \vcknapsack is strongly \NPC, that is, it is \NPC even if the weight and value of every vertex are encoded in unary. We reduce from the classical \vc problem, where the goal is to find a vertex cover of cardinality at most some input integer $k$. \vc is known to be \NPC even for $3$ regular graphs~\cite[folklore]{DBLP:journals/dm/FleischnerSS10}. To reduce a \vc instance to a \vcknapsack instance, we define the weight and value of every vertex to be $1$, and the size and target value to be $k$. 
Some of our observations and theorems are intuitive. We omit the proofs of a few of our results. They are marked $(\star)$.

\begin{observation}\label{thm:vckp-gen-npc}
	\vcknapsack is strongly \NPC.
\end{observation}
 \begin{proof}
     Clearly, \vcknapsack is in \NP. We reduce \vc to \vcknapsack to prove NP-hardness. Let $(\GG(\VV=\{v_i: i\in[n]\},\EE),k)$ be an arbitrary instance of \vc. We construct the following instance $(\GG^\pr(\VV^\pr=\{u_i: i\in[n]\},\EE^\pr),(w(u))_{u\in\VV^\pr},(\alpha(u))_{u\in\VV^\pr},s,p)$ of \vcknapsack.
	\begin{align*}
		&\VV^\pr = \{u_i : v_i\in \VV, \forall i\in[n]\}\\
		&\EE^\pr = \{\{u_i,u_j\}: \{v_i, v_j\} \in\EE, i\neq j, \forall i, j\in[n]\}\\
		&w(u_i) = 1, \alpha(u_i)=1 \qquad \forall i\in[n]\\
		&s =p = k
	\end{align*}
	The \vcknapsack problem has a solution iff \vc has a solution.\\
	
	\noindent
	Let $(\GG^\pr,(w(u))_{u\in\VV^\pr},(\alpha(u))_{u\in\VV^\pr},s,p)$ of \vcknapsack such that $\WW^\pr$ be the resulting subset of $\VV^\pr$ with (i) $\WW^\pr$ is a \vc, (ii) $\sum_{u\in \WW^\pr} w(u) = k$, (iii) $\sum_{u\in \WW^\pr} \alpha(u) = k$\\
	\noindent
	This means that the set $\WW^\pr$ is a \vc which gives the maximum profit $k$ for the bag capacity of size $k$. In other words, $\WW^\pr$ is a \vc of size $k$. Since $\WW' =\{u_i: v_i \in \WW, \forall i \in[n]\}$, $\WW$ is a \vc of size $k$. Therefore, the \vc instance is an \yes instance.\\
	
	\noindent
	Conversely, let us assume that \vc instance $(\GG,k)$ is an \yes instance.\\ 
	Then there exists a subset $\WW\subseteq \VV$  of size $k$ such that it outputs a \vc. \\	
	\noindent
	Consider the set $\WW^\pr=\{u_i: v_i \in \WW, \forall i \in[n]\}$.
	Since each vertex of $\WW^\pr$ is involved with weight 1 and produces profit amount 1, $\WW^\pr$ is a \vc of max bag size and total profit $k$. 
	\noindent
	Therefore, the \vcknapsack instance is an \yes instance.
 \end{proof}

The same reduction in \Cref{thm:vckp-gen-npc} also shows that \vcknapsackbudget is strongly \NPC.

\begin{corollary}\label{thm:vckb-gen-npc}
	\vcknapsackbudget is strongly \NPC.
\end{corollary}

In the \maxminimalvc problem, the goal is to compute if there exists a minimal vertex cover of cardinality at least some input integer. A vertex cover of a graph is called minimal if no proper subset of it is a vertex cover. \maxminimalvc is known to be \NPC \cite{DBLP:conf/aaim/HanIKL07,DBLP:journals/dam/BoriaCP15}. We show that the same reduction as in the proof of \Cref{thm:vckp-gen-npc} except starting from an instance of \maxminimalvc instead of \vc, proves that \minimalvcknapsack is strongly \NPC.

\begin{observation}\label{thm:minimalvckp-gen-npc}
	\minimalvcknapsack is strongly \NPC.
\end{observation}
\begin{proof}
	Clearly, \minimalvcknapsack $\in$ \NP. Since the decision version of \maxminimalvc is \NPC,  we reduce \maxminimalvc to \minimalvcknapsack to prove NP-completeness. Let $(\GG(\VV=\{v_i: i\in[n]\},\EE),k)$ be an arbitrary instance of \maxminimalvc of size $k$. We construct the following instance $(\GG^\pr(\VV^\pr=\{u_i: i\in[n]\},\EE^\pr),(w(u))_{u\in\VV^\pr},(\alpha(u))_{u\in\VV^\pr},s,p)$ of \minimalvcknapsack.
	\begin{align*}
		&\VV^\pr = \{u_i : v_i\in \VV, \forall i\in[n]\}\\
		&\EE^\pr = \{\{u_i,u_j\}: \{v_i, v_j\} \in\EE, i\neq j, \forall i, j\in[n]\}\\
		&w(u_i) = 1, w(u_j)=0 \quad \text{for~some~} i,j \in [1,n] \\
		&\alpha(u_i)=1 \quad \forall i \in [1,n] \\
		&s = p =k
	\end{align*}
	The reduction of \maxminimalvc to \minimalvcknapsack works in polynomial time. The \minimalvcknapsack problem has a solution iff \maxminimalvc has a solution.\\
	\noindent
	Let $(\GG^\pr(\VV^\pr,\EE^\pr),(w(u))_{u\in\VV^\pr},(\alpha(u))_{u\in\VV^\pr},s,p)$ of \minimalvcknapsack such that $\WW$ be the resulting subset of $\VV'$ with (i) $\WW$ is a minimal \vc, (ii) $\sum_{u\in \WW} w(u) \leq s = k$, (iii) $\sum_{u\in \WW} \alpha(u) \geq p = k$.\\
	From the above condition (i), as $\WW$ is minimal \vc we can not remove any vertex from $\WW$ which gives another \vc. Now, since we rewrite the second and third conditions as $|\WW| \leq k$ and $|\WW| \geq k$. So, $|\WW|$ must be $k$. Now, let $\II=\{v_i:u_i \in \WW,\forall i\in[n]\}$. Since $\WW$ is the Minimal \vc of size $k$, then $\II(\subseteq \VV)$ must be \maxminimalvc set of size at least $k$.Therefore, the \maxminimalvc instance is an \yes instance.\\
	
	\noindent
	Conversely, let us assume that \maxminimalvc instance $(\GG,k)$ is an \yes instance.\\ 
	Then there exists a subset $\II\subseteq \VV$  of size $ \geq k$ such that it outputs an \maxminimalvc. Let $\WW=\{u_i:v_i \in \II,\forall i\in[n]\}$. Since $\II$ is an \maxminimalvc of size $ \geq k$, then $\WW$ must be a Minimal \vc of size $k$. So, $\sum_{u\in \WW} w(u) \leq k =s$ and $\sum_{u\in \WW} \alpha(u) \geq k = p$. Now, we want to prove that $\WW$ is a minimal \vc.\\
	\noindent
	Therefore, the \minimalvcknapsack instance is an \yes instance.
\end{proof}

We show similar results for \minimumvcknapsack except that it does not belong to \NP unless the polynomial hierarchy collapses.

\begin{observation}\label{minimumvck-nph}
    \minimumvcknapsack is strongly \NPH.
\end{observation}

\begin{proof}
Clearly, \minimumvcknapsack does not belong to the class \NP. To prove hardness, we reduce \minimumvcknapsack from \vc. Let $(\GG(\VV=\{v_i: i\in[n]\},\EE),k)$ be an arbitrary instance of \vc, where we ask if there is a vertex cover of size at most $k$. We construct the following instance $(\GG^\pr(\VV^\pr=\{u_i: i\in[n]\},\EE^\pr),(w(u))_{u\in\VV^\pr},(\alpha(u))_{u\in\VV^\pr},s,p)$ of \vcknapsack.
	\begin{align*}
		&\VV^\pr = \{u_i : v_i\in \VV, \forall i\in[n]\}\\
		&\EE^\pr = \{\{u_i,u_j\}: \{v_i, v_j\} \in\EE, i\neq j, \forall i, j\in[n]\}\\
		&w(u_i) = 1, \alpha(u_i)=0 \qquad \forall i\in[n]\\
		&s = k\\
            &p = 0
	\end{align*}
	The \minimumvcknapsack problem has a solution iff \vc has a solution.\\ Let us assume $\WW \subseteq \VV$ be the solution to \vc.
Let $(\GG^\pr,(w(u))_{u\in\VV^\pr},(\alpha(u))_{u\in\VV^\pr},s,p)$ of \minimumvcknapsack such that $\WW^\pr$ be the resulting subset of $\VV^\pr$ with (i) $\WW^\pr$ is a minimum \vc, (ii) $\sum_{u\in \WW^\pr} w(u) = k$, (iii) $\sum_{u\in \WW^\pr} \alpha(u) = 0$\\
	\noindent
	This means that the set $\WW^\pr$ is a \vc which gives the maximum value 0 for the bag capacity of size $k$. In other words, $\WW^\pr$ is a \vc of size at most $k$. Since $\WW' =\{u_i: v_i \in \WW, \forall i \in[n]\}$, $\WW$ is a \vc of size $k$. Therefore, the \vc instance is an \yes instance.\\
	
	\noindent
	Conversely, let us assume that \vc instance $(\GG,k)$ is an \yes instance.\\ 
	Then there exists a subset $\WW\subseteq \VV$  of size at most $k$ such that it outputs a \vc. This means there exists a minimum vertex cover of size at most $k$.\\	
	\noindent
	Consider the set $\WW^\pr=\{u_i: v_i \in \WW, \forall i \in[n]\}$.
	Since each vertex of $\WW^\pr$ is involved with weight 1 and produces value amount 0, $\WW^\pr$ is a \vc of max bag size i.e. $k$ and total profit $0$. 
	\noindent
	Therefore, the \minimumvcknapsack instance is also an \yes instance.
 \end{proof}

We show next that \vcknapsack is \NPC even if the underlying graph is a tree by reducing it from the classical \kp --- simply add the knapsack items as leaves of a star graph. However, it turns out that they are not strongly \NPC for trees. We will see in \Cref{parameterized} that they admit pseudo-polynomial time algorithms for trees.

\begin{observation}($\star$)\label{vck-trees-npc}
\vcknapsack is \NPC for star graphs.
\end{observation}

By setting $k$ to be the number of leaves, the reduction in the proof of \Cref{vck-trees-npc} also shows \NP-completeness for \vcknapsackbudget.

\begin{observation}($\star$)\label{k-vck-trees-npc}
\vcknapsackbudget is \NPC for star graphs.
\end{observation}

Note that the reduction from \kp to \vcknapsack for star graphs does not work for \minimalvcknapsack and \minimumvcknapsack. Indeed, for star graphs, both the problems admit polynomial-time algorithms. Nevertheless, we are able to show that both the problems are (not strongly) \NPC for trees.

\begin{theorem}\label{minimalvck-trees-npc}
\minimalvcknapsack is \NPC for trees.
\end{theorem}

\begin{proof}
Clearly, \minimalvcknapsack for tree $\in$ NP. We reduce \minimalvcknapsack from \kp to prove hardness. Given, an arbitrary instance $(\II,(\theta_i)_{i\in\II}, (p_i)_{i\in\II}, b,q)$ of \kp,  we create an instance of \minimalvcknapsack for trees by forming a complete binary tree \TT = (\VV, \EE) that has $n$ nodes in its penultimate level and denote an instance as $(\TT,(w(u))_{u\in\TT},(\alpha(u))_{u\in\TT},s,p)$. The construction steps are as follows:

\begin{enumerate}
    \item We construct a complete binary tree \TT such that the penultimate level of this tree has $n$ nodes representing the $n$ items.
    \item Without loss of generality we assume, that $n=|\II|=2^k$ because suppose there exists $n$ such that $2^{k-1} < n < 2^{k} $, add  $2^{k} - n$ items to \II each with size = 0 and profit = 0, where $k \ge 0$. 
    \item Let $\VV^\pr \subseteq \TT[\VV]$ denote the vertices belonging to the penultimate level of the tree \TT.
    \item Set $w(u_i)_{u_i \in \VV^\pr} = \theta_i, \alpha(u_i)_{u_i \in \VV^\pr} = p_i, 
			w(u_i)_{u_i \in \VV - \VV^\pr}=0, \alpha(u_i)_{u_i \in \VV - \VV^\pr}=0 \\
			s=b, p=q$
\end{enumerate}

This construction takes polynomial time.
\par We now show that \kp $\le_p$ \minimalvcknapsack.

The \minimalvcknapsack problem has a solution if and only if \kp has a solution.

Let us consider $(\TT,(w(u))_{u\in\VV},(\alpha(u))_{u\in\VV},s,p)$ be an \yes instance of \minimalvcknapsack such that $\WW$ be the resulting subset of $\VV$ with $\WW$ is a minimal \vc, $\sum_{u\in \WW} w(u) = b$, and $\sum_{u\in \WW} \alpha(u) = p$. This means that the set $\WW$ is a minimal \vc, which gives the maximum profit $p$ for the bag capacity of size $b$. In other words, $\WW \cap \VV^\pr = \WW^\pr \ne \phi$ i.e. \WW must contain at least a subset of $\VV^\pr$ in order to attain the required profit $p$ because \WW - \VV' has vertices whose weight = 0, and profit = 0. As per our construction, \VV' contains the vertices denoting the $n$ items of \II in \kp.
	\noindent
	Therefore, the \kp instance is an \yes instance.\\

Conversely, let us assume that $(\II,(\theta_i)_{i\in\II}, (p_i)_{i\in\II}, b,q)$ be an \yes instance of \kp such that $\II$ be the resulting subset of $\II$ with $\sum_{i\in \II} \theta_i \leq b$, and  $\sum_{i\in \II} p_i \geq q$.	The minimal vertex cover knapsack must pick vertices from $\VV^\pr$ otherwise it cannot meet the desired value $p$ as all vertices of $\VV - \VV^\pr$ have weight and value equal to 0.

 We obtain a minimal \vc for \TT by recursively constructing the solution set, \WW as:
\begin{enumerate}
     \item if $u \in \TT[\II]$, include $u$ in \WW and delete $\NN[u]$ from \TT
     \item if $u \notin \TT[\II]$, include $\NN(u$ in \WW delete $\NN[u]$ from \TT 
     \item solve for minimal vertex cover on the induced sub-graph where weight and profit associated with each vertex is 0.   
 \end{enumerate}

\noindent
	Therefore, the \minimalvcknapsack for tree is an \yes instance.
\end{proof}

\begin{observation}\label{minimumvck-trees-npc}
\minimumvcknapsack is \NPC for trees.
\end{observation}
\begin{proof}
Clearly, \minimumvcknapsack for trees belongs to \NP. We reduce \minimumvcknapsack from \kp to prove hardness. Given, an arbitrary instance $(\II,(\theta_i)_{i\in\II}, (p_i)_{i\in\II}, b,q)$ of \kp,  we create an instance of \minimumvcknapsack for trees and denote it as $(\TT,(w(u))_{u\in\TT},(\alpha(u))_{u\in\TT},s,p)$. The construction steps are as follows: 
\begin{enumerate}
    \item Create $m, (m>1)$ vertices each with weight = 0 and value = 0.
    \item Create another vertex with weight = 0 and value = 0. Attach each of the vertices created in Step 1 to this vertex.
    \item Create $u_1, u_2, \dots,u_{|\II|}$ vertices. Set \(w(u_i) = \theta_i, \forall i \in \II\) and \(\alpha(u_i) = p_i, , \forall i \in \II\). Attach each of these vertices to the vertex created in Step 2.
    \item Create $v_1, v_2, \dots,v_{|\II|}$ vertices. Set \(w(v_i) = 0, \forall i \in \II\) and \(\alpha(v_i) = 0 , \forall i \in \II\). Connect \(\{u_i,v_i\}, \forall i\in [|\II|]\) with an edge $e_i$.
    \item Set $s=b$ and $p=q$
\end{enumerate}

We claim that the constructed instance is an \yes instance for \minimumvcknapsack for trees iff \kp is an instance.

Let us assume that \minimalvcknapsack for trees is an \yes instance and let $\WW^\pr \subseteq \TT$ be its solution. This means that the vertex created in Step 2 is a forced vertex (otherwise the solution $\WW^\pr$ is not minimum, a contradiction). Now to cover each of the edges \(\{u_i,v_i\}, \forall i\in [\II]\), we have a choice of two vertices i.e. either $u_i$ or $v_i$. This selection purely translates to the problem of solving the knapsack constraint on the vertex set $\{u_i\}_{i \in [|\II|]}$ otherwise we cannot meet the desired total value $p$ within the knapsack size $s$. Thus the \kp problem is also an \yes instance.

Conversely, let \kp be ab \yes instance and let 
$\WW^\pr \subseteq \II$ be its solution. This means we have a solution to \kp problem where the total value is at least $p$ and total value is at most $b$. Therefore, in the \minimumvcknapsack instance we choose the vertices $u_i \in \TT$ corresponding to the item $\theta_i \in \WW$ to cover edge $e_i$. If there exists some edge $e_i$ for which its corresponding $u_i$ is not in \WW, we choose $v_i$ to cover such edge(s). Interestingly, both the weight and value of $v_i \forall i \in [|\II|]$ is 0. Hence, there exists a minimum vertex cover of total weight at most $s$ and total value at least $p$ in the \minimumvcknapsack instance. Therefore, the \minimumvcknapsack instance is also an \yes instance.
\end{proof}

Note that, since the size of a minimum vertex cover in a tree can be computed in polynomial time thanks to K\"{o}nig's Theorem~\cite{west2001introduction}, \minimumvcknapsack belongs to \NP.

We show similar results for \setck and \hsk also by reducing from respectively unweighted set cover and unweighted $d$-hitting set, both of which are known to be \NPC~\cite{DBLP:books/fm/GareyJ79}.

\begin{observation}($\star$)\label{setcover-knapsack-npc}
   \setck is strongly \NPC. 
\end{observation}

\begin{observation}($\star$)\label{hitting-knapsack-npc}
   \hsk is strongly \NPC. 
\end{observation}

%% file: results-poly-approx.tex
\section{Results: Polynomial Time Approximation Algorithms}
\label{approximation}

In this section, we focus on the polynomial-time approximability of our problems. For all the problems in this paper, we study two natural optimization versions: (i) minimizing the weight of the solution given a target value as input and (ii) maximizing the value of the solution given knapsack size as input. We first consider minimizing the weight of the solution.

A natural integer linear programming formulation of \vcknapsack is the following.

minimize $\sum_{u \in \VV} w(u)x_u$ 
\begin{align*}
    \text{Subject to:}\\
       &x_u + x_v \geq 1,\forall (u,v) \in \EE \\\numberthis \label{ILP-sc}
        &\sum_{u\in \VV} \alpha(u)x_u \geq p \\
    &x_u \in \{0,1\}, \forall u \in \VV\numberthis \label{ILP-vc}
\end{align*}

We replace the constraints $x_u \in \{0,1\}$, with $x_u \geq 0$, $\forall u\in \VV$ to obtain linear programming (abbreviated as LP) relaxation of the integer linear program (abbreviated as ILP).
\begin{observation}\label{integrality-gap}
    The relaxed LP of the ILP \ref{ILP-vc} has an unbounded integrality gap. To see this, consider an edgeless graph on two vertices $v_1$ and $v_2$. Let $w(v_1) = 0$, $w(v_2) = 1$, $\alpha(v_1) = p-1$ and $\alpha(v_2) = p$. The optimal solution to ILP sets $x_{v_1} = 0, x_{v_2} = 1$, for a total weight of 1. However, the optimal solution to the relaxed LP sets $x_{v_1} = 1, x_{v_2} = 1/p$ and has a total weight of $1/p$. Thus, in this case, the integrality gap is at least $\frac{1}{1/p} = p$. 
\end{observation}

To tackle \Cref{integrality-gap}, we strengthen the inequality $\sum_{u\in \VV} \alpha(u)x_u \geq p$. This allows us to obtain an $f$ approximation algorithm even for the more general \setck problem that we present now. In particular, in addition to having a constraint for every element of the universe, we have a constraint for every $\AA \subseteq \FF$ of sets such that $\alpha(A) = \sum_{i \in \AA} \alpha(i) < p$ where $p$ is the target value given as input. We define the residual value $p_\AA = p - \alpha(\AA)$. Given the set \AA, we simplify the problem on the sets $\FF-\AA$, where the target value is now $p_\AA$. We also reduce the value of each set $S_i \in \FF-\AA$ to be the minimum of its own value and $p_\AA$, i.e., let $\alpha^\AA(i)$ = min($\alpha(i)$, $p_\AA$). We can now give the following Integer linear programming formulation of the problem:


minimize $\sum_{i\in [m]} w(i)x_i$ 
\begin{align*}
    \text{Subject to:}\\
       &\sum_{i:e_j \in \SS_i}x_i \geq 1, \forall e_j \in \UU \\
        &\sum_{i\in \FF-\AA} \alpha^\AA(i)x_i \geq p_\AA, \forall \AA \subseteq \FF \\
    &x_i \in \{0,1\}, \forall i \in [m]
\end{align*}
We replace the constraints $x_i \in \{0,1\}$, with $x_i \geq 0$ to obtain the LP relaxation of the ILP. The dual of the LP relaxation is :

maximize $\sum_{\AA: \AA \subseteq \FF} p_\AA y_\AA$ + $\sum_{j\in[n]}y_j$
\begin{align*}
    \text{Subject to:}\\
       &\sum_{j:e_j \in \SS_i}y_j \leq w(i), \forall \SS_i \in \FF \\
        &\sum_{\AA\subseteq\FF: i \not\in \AA} \alpha^\AA(i)y_\AA \leq w(i), \forall i \in \FF \\
    &y_\AA \geq 0, \forall \AA \subset \FF
\end{align*}

In our primal-dual algorithm, we begin with dual feasible solution $y=0$ and partial solution $\AA=\emptyset$. We pick one set in every iteration until the value of the set \AA of sets becomes at least the target value $p$. We increase the dual variable $y_\AA$ in every iteration until the dual constraint for any set $i\in \FF-\AA$ becomes tight. We then pick that set in our solution and continue. After this loop terminates, the value of the set \AA of sets is at least the target value $p$. At that point, if \AA is a set cover, then we output \AA. Otherwise, till there exists an element $e_j$ of the universe that is not covered by \AA, we increase the dual variable $y_j$ until the dual constraint for some set \el with $e_j\in S_l$ becomes tight. We then include $S_\el$ in \AA and continue. We present the pseudocode of our algorithm in \Cref{pdsc}.

\begin{algorithm}
\caption{Primal-dual $f$-approximation algorithm for \setck}
\begin{algorithmic}[1]
\longversion{\State $y \gets 0$
\State $\AA \gets \emptyset$}
\shortversion{\State $y \gets 0, \AA \gets \emptyset$}
\While{$\alpha(\AA) < p$}
    \State Increase $ y_\AA$ until for some $i \in \FF - \AA, \sum_{\BB\subseteq\FF: i \not\in \BB} \alpha^\BB(i)y_\BB = w(i)$
    \State $\AA \gets \AA \cup \{i\}$
\EndWhile
\longversion{\State $\XX \gets \AA$
\State $\AA^\pr \gets \AA$}
\shortversion{\State $\XX \gets \AA, \AA^\pr \gets \AA$}
\While{$\exists e_j \not \in \bigcup_{i \in \AA^\pr} S_i$}
    \State Increase $y_j$ until there is some $t$ with $e_j \in S_t$ such that $\sum_{j:e_j \in S_t}y_j = w(t) $
    \State $\AA^\pr \gets \AA \cup \{t\}$
\EndWhile
\State \textbf{return} $\AA^\pr$
\end{algorithmic}
\label{pdsc}
\end{algorithm}

\begin{theorem}\label{thm:approxf}
\Cref{pdsc} is an $f$-approximation algorithm for the \setck problem for minimizing the weight of the solution, where $f$ is the maximum number of sets in the family where any element belongs.
\end{theorem}

\begin{proof}

Let \ALG be the weight of the set cover $\AA^\pr$ output by \Cref{pdsc}. Then 
\begin{align*}
    \ALG &= \sum_{i\in \AA^\pr} w(i)x_i \\
        &= \sum_{i\in \XX} w(i)x_i + \sum_{i\in \mathcal{A}^\prime - \mathcal{X}} w(i)x_i 
\end{align*}

Let \OPT be the optimal weight of the \setck instance, set $i$ picked in the $i$-th iteration of the first while loop (which we can assume without loss of generality by renaming the sets), and $l$ the set selected by the algorithm at the last iteration of the first while loop. Since the first while loop terminates when $\alpha(\AA) \geq p$, we know that $\alpha(\XX) \geq p$; since set $l$ was added to $\XX$, it must be the case that before $l$ was added, the total value of \AA was less than $p$, so that $\alpha(\XX - \{l\}) < p$. For $i\in[l]$, we define $\AA_i=[i]$ as the set of sets picked in the first $i$ iterations of the first while loop and $\CC=\{\AA_i: i\in[l]\}$. We observe that a dual variable $y_\BB$ is non-zero only if $\BB\in\CC$. Since we pick only tight sets, we have
\begin{equation*}
   \sum_{i \in \XX} w(i) = \sum_{i \in \XX} \sum_{\BB \subseteq \FF: i \notin \BB} \alpha^\BB(i) y_\BB = \sum_{i \in \XX} \sum_{\BB \in \CC: i \notin \BB} \alpha^\BB(i) y_\BB. 
\end{equation*}

Reversing the double sum, we have
\begin{equation*}
\sum_{i \in \XX} \sum_{\BB \in \CC: i \notin \BB} \alpha^\BB(i) y_\BB = \sum_{\BB \in \CC} y_\BB \sum_{i \in \XX- \BB} \alpha^\BB(i).
\end{equation*}

Note that in any iteration of the algorithm except the last one, adding the next set $i$ to the current sets in $\AA$ did not cause the value of the knapsack to become at least $p$; that is, $\alpha(i) < p - \alpha(\AA) = p_\AA$ at that point in the algorithm. Thus, for all sets $i \in \AA$ except $l$, $\alpha^\AA(i) = \min(\alpha(i), p_\AA) = \alpha(i)$, for the point in the algorithm at which $\AA$ was the current set of sets. Thus, we can rewrite
\begin{equation*}
\sum_{i \in \XX-\AA}  \alpha^\AA(i) =  \alpha^\AA(l) + \sum_{i \in \XX - \AA : i \neq l} \alpha^\AA(i) = \alpha^\AA(l) + \alpha(\XX-\{l\}) - \alpha(\AA).
\end{equation*}

Note that $\alpha^\AA(l)  \leq p_\AA$ by definition, and as argued at the beginning of the proof $\alpha(\XX-\{l\}) < p$ so that $\alpha(\XX-\{l\})  - \alpha(\AA) < p - \alpha(\AA)) = p_\AA$; thus, we have that
\begin{equation*}
\alpha^\AA(l) + \alpha(\XX-\{l\}) - \alpha(\AA) < 2p_\AA
\end{equation*}
which is the same as saying
\[\sum_{i \in \XX- \BB} \alpha^\BB(i)< 2 p_\BB \text{ for every }\BB\in\CC.\]
Therefore,
\begin{equation*}
\sum_{i \in \XX} w(i) = \sum_{\BB \in \CC} y_\BB \sum_{i \in \XX- \BB} \alpha^\BB(i)< 2 \sum_{\BB : \BB \in \CC} p_\BB y_\BB = 2\sum_{\BB \subseteq \FF: i \notin \BB} p_\BB y_\BB
\end{equation*}
where the last equality follows from the fact that $y_\BB=0$ if $\BB\notin\CC$.

Our algorithm picks sets in $\AA\pr\setminus\XX$ in the second while loop if the set of sets picked in the first while loop does not form a set cover. We now upper bound $\sum_{i\in\AA^\pr\setminus\XX}w(i)$ as follows.
\begin{align*}
    \sum_{i\in\AA^\pr\setminus\XX}w(i) = \sum_{i\in\AA^\pr\setminus\XX} \sum_{j\in[n]:e_j\in S_i} y_j = \sum_{j\in[n]}|\{i\in\AA^\pr\setminus\XX: e_j\in S_i\}|y_j \le f \sum_{j\in[n]}y_j
\end{align*}
The first equality follows from the fact that only tight sets are picked. We now bound \ALG.
\begin{align*}
    \ALG &= \sum_{i\in \AA^\pr} w(i)x_i \\
        &= \sum_{i\in \XX} w(i)x_i + \sum_{i\in \AA^\pr - \XX} w(i)x_i \\
        &\leq 2 \sum_{A : A \subseteq I} p_\AA y_\AA + f\sum_{j\in[n]}y_j\\
       &\leq f \left( \sum_{A : A \subseteq I} p_\AA y_\AA + \sum_{j\in[n]}y_j \right)\\
        &= f\text{\OPT}\qedhere
\end{align*}
\end{proof}

We note that our algorithm is a combinatorial algorithm based on the primal-dual framework --- in particular, we use LPs only to design and analyze our algorithm. We do not need to solve any LP. We obtain approximation algorithms for the \vcknapsack and \hsk problems as corollaries of \Cref{thm:approxf} by reducing these problems to \setck.

\begin{corollary}($\star$)\label{hsdfactor}
There exists a \textbf{d}-approximation algorithm for \hsk for minimizing the weight of the solution. The algorithm is combinatorial in nature and based on the primal-dual method.
\end{corollary}

\begin{corollary}($\star$)\label{2factorvc}
There exists a $2$-approximation algorithm for \vcknapsack for minimizing the weight of the solution. The algorithm is combinatorial in nature and based on the primal-dual method. 
\end{corollary}

We next present a $H_d$-approximation algorithm for \setck where $d$ is the maximum cardinality of any set in the input instance and $H_d=\sum_{i=1}^d \frac{1}{i}$ is the $d$-th harmonic number. The idea is to run the first while loop of \Cref{pdsc}, and then, if the selected sets do not cover the universe, then, instead of the second while loop of \Cref{pdsc}, we pick sets following the standard greedy algorithm for minimum weight set cover. We show that the algorithm achieves an approximation factor of max$(2,H_d)$ by analyzing it using the {\em dual fitting technique}.

\begin{algorithm}[htb]
    \caption{Max(2, $H_d$)-approximation algorithm for \setck}
    \begin{algorithmic}[1]
        \longversion{\State $y \gets 0$
        \State $\AA \gets \emptyset$}
        \shortversion{\State $y \gets 0, \AA \gets \emptyset$}
        \While{$\alpha(\AA) < p$}
             \State Increase $ y_\AA$ until for some $i \in \FF - \AA, \sum_{\BB\subseteq\FF: i \not\in \BB} \alpha^\BB(i)y_\BB = w(i)$
            \State $\AA \gets \AA \cup \{i\}$
        \EndWhile
        \longversion{\State $\XX \gets \AA$
        \State $\UU^\pr \gets \UU-\bigcup_{i \in \XX} \SS_i $
        \State $\FF^\pr \gets \FF-\XX$
        \State $I \gets \emptyset$
        \State $\hat{S}_i \gets S_i$ for all $i \in \FF^\pr$}
        \shortversion{\State $\XX \gets \AA, \UU^\pr \gets \UU-\bigcup_{i \in \XX} \SS_i, \FF^\pr \gets \FF-\XX, I \gets \emptyset, \hat{S}_i \gets S_i$ for all $i \in \FF^\pr$}
        \While{$I$ is not a set cover for $\UU^\pr$}
            \State $l \gets \text{arg min}_{i:\hat{S}_i \neq \emptyset} \frac{w(i)}{|\hat{S}_i|}$
            \State $I \gets I \cup \{l\}$
            \State $\hat{S}_i \gets \hat{S}_i - S_l$ for all $i\in \FF^\pr$
        \EndWhile
        \State \textbf{return} $\XX \cup \II$
    \end{algorithmic}
    \label{lognapprox}
\end{algorithm}

\begin{theorem}\label{hgapprosetcover}
\Cref{lognapprox} is a max(2, $H_d$)-approximation algorithm for the \setck problem for minimizing the weight of the solution, where $d$ is the maximum cardinality of any set in the input.
\end{theorem}

\begin{proof}
We follow the same notation defined in \Cref{lognapprox} in this proof. Since the first part of \Cref{lognapprox} is the same as the first part of \Cref{pdsc}, from the proof of \Cref{thm:approxf}, we have
\begin{equation*}
\sum_{i \in \XX} w(i) < 2\sum_{\BB \subseteq \FF: i \notin \BB} p_\BB y_\BB.
\end{equation*}

To bound the sum of weights of the sets in \II, we use the dual fitting technique. In particular, we will first construct an assignment of dual variables $(y_j)_{j\in[n]}$ with $\sum_{i \in \II} w_i = \sum_{j=1}^n y_j$. However, $(y_j)_{j\in[n]}$ may not satisfy the dual constraints involving those variables. However, and then show that \( y^\pr_j = \frac{1}{H_d} y_j, j\in[n] \) satisfies all dual constraints involving those variables. We concretize this idea below.

Whenever \Cref{lognapprox} includes a set $\hat{\SS}_i$ in \II, we define \( y_j = \frac{w(i)}{|\hat{\SS}_i|} \) for each \( j \in \hat{\SS}_i \). Since each \( j \in \hat{\SS}_i \) is uncovered in iteration when \Cref{lognapprox} picks the set $\hat{\SS}_i$, and is then covered for the remaining iterations of the algorithm (because we added subset \( \SS_i \) to \II), the dual variable \( y_j \) is set to a value exactly once. Furthermore, we see that
\begin{equation*}
w(i) = \sum_{i: j \in \hat{\SS}_i} y_j, \forall i\in\II
\end{equation*}
since the weight $w(i)$ of the set $i$ is distributed among $y_j, j\in\hat{\SS}_i$. Hence, we have,
\begin{equation*}
\sum_{j \in \II} w(i) = \sum_{i=1}^n y_j.
\end{equation*}

We claim that \( y'_j = \frac{1}{H_d} y_j \) for all $j\in[n]$ satisfies the dual constraints involving these variables. For that, we need to show that for each subset \( \SS_i, i\in[m] \),
\begin{equation*}
\sum_{i: j \in \SS_i} y'_j \leq w(i).
\end{equation*}
Pick an arbitrary subset \( \SS_i \) and an arbitrary iteration $k$ of the second while loop of \Cref{lognapprox}. Let \el be the number of iterations that the second while loop of \Cref{lognapprox} makes and \( a_k \) the number of elements
in this subset that is still uncovered at the beginning of the \( k \)-th iteration, so that \( a_1 = |\SS_i| \),
and \( a_{\ell+1} = 0 \). Let \( A_k \) be the set of uncovered elements of \( \SS_i \) covered in the \( k \)-th iteration, so that
\(|A_k| = a_k - a_{k+1}\). If subset \( \SS_q \) is chosen in the \( k \)-th iteration, then for each element \( j \in A_k \)
covered in the \( k \)-th iteration, we have
\begin{equation*}
y'_j = \frac{w_q}{H_d |\hat{\SS}_q|} \leq \frac{w(i)}{H_d a_k},
\end{equation*}
where \( \hat{\SS}_q \) is the set of uncovered elements of \( \SS_q \) at the beginning of the \( k \)-th iteration. The inequality follows because if \( \SS_q \) is chosen in the \( k \)-th iteration, it must minimize the ratio of its
weight to the number of uncovered elements it contains. Thus,
\begin{align*}
\sum_{i: e_j \in \SS_i} y'_j &= \sum_{k=1}^l \sum_{j\in[n]: j \in A_k} y'_j \\
&\leq \sum_{k=1}^l (a_k - a_{k+1}) \frac{w(i)}{H_d a_k} \\
&\leq \frac{w(i)}{H_d} \sum_{k=1}^l \frac{a_k - a_{k+1}}{a_k} \\
&\leq \frac{w(i)}{H_d} \sum_{k=1}^l \left( \frac{1}{a_k} + \frac{1}{a_k-1} + \cdots + \frac{1}{a_{k+1} + 1} \right) \\
&\leq \frac{w(i)}{H_d} \sum_{i=1}^{|\SS_i|} \frac{1}{i} \\
&= \frac{w(i)}{H_d} H_{|\SS_i|} \\
&\leq w(i),
\end{align*}
where the final inequality follows because \( |\SS_i| \leq d \). Hence, $((y_\BB)_{\BB\in\FF},(y^\pr_j)_{j\in[n]})$ is a dual feasible solution. We now bound \ALG as follows. 
\begin{align*}
    \ALG  &= \sum_{i\in \XX} w(i)x_i + \sum_{i\in \II} w(i)x_i\\
    &\leq 2 \sum_{A : A \subseteq \FF} p_\AA y_\AA  + H_d \sum_{j\in[n]}y_j \\  
    &= \max(2, H_d) \left( \sum_{A : A \subseteq \FF} p_\AA y_\AA  + \sum_{j\in[n]}y_j \right)\\
    &= \max(2, H_d)\cdot \text{\OPT}\qedhere
\end{align*}
\end{proof}

The approximation guarantees of \Cref{hgapprosetcover,thm:approxf} are the best possible approximation guarantees, up to any additive constant $\eps>0$, that any polynomial time algorithm hopes to achieve, assuming standard complexity-theoretic assumptions.

\begin{theorem}\label{approx-lb}
    Let $\eps>0$ be any constant. Then we have the following:
    \begin{enumerate}
        \item There is no polynomial-time $(1-\eps)\ln n$ factor approximation algorithm for \setck unless every problem in \NP admits a quasi-polynomial time algorithm.

        \item Assuming Unique Games Conjecture (UGC), there is no polynomial-time $(1-f)\ln n$ factor approximation algorithm for \setck.

        \item Assuming Unique Games Conjecture (UGC), there is no polynomial-time $(1-d)\ln n$ factor approximation algorithm for \hsk.
    \end{enumerate}
\end{theorem}

\begin{proof}
    We present an approximation preserving reduction from \setc to \setck. Let $([n],\FF,k)$ be an arbitrary instance of \setc. We construct an instance of \setck where the universe is $[n]$, the family of sets is \FF, the weight and value of every set in \FF is $1$, the bag size and the target value both are $k$. Clearly, the \setc instance is a \yes instance if and only if the \setck instance is a \yes instance. Moreover, if there exists a polynomial-time $\alpha$-approximation algorithm for minimizing the weight of the solution of \setck, then there exists a polynomial-time $\alpha$-approximation algorithm for minimizing the weight of \setc. Now, the results on \setck follow from $(1-\eps)\ln n$ and $(1-f)\ln n$ lower bounds for \setc under respective assumptions~\cite{DBLP:journals/jacm/Feige98,DBLP:journals/jcss/KhotR08}. For \hsk also, a similar reduction and the known lower bound on the approximability of \hs proves the result.
\end{proof}

We next focus on maximizing the value of the solution given a knapsack size as input. Surprisingly, for all the problems studied in this paper, we show that there is no $\rho$-approximation algorithm for any of our problems for any $\rho>1$.

\begin{theorem}\label{hardnessofapprox}
For any $\rho > 1$, there does not exist a $\rho$-approximation algorithm for maximizing the value of the solution given the size of the knapsack for \setck, \hsk, \vcknapsack, \minimalvcknapsack, \minimumvcknapsack, and \vcknapsackbudget unless $\Pb = \NP$.
\end{theorem}

\begin{proof}
We prove the result for \setck. For other problems, the proof is similar. Suppose there exists a $\rho$-approximation algorithm \AA for \setck for maximizing the value of the solution. We show that we can design a polynomial time algorithm for unweighted \setc using \AA, which is known to be \NPC. Let $([n],\FF,k)$ be an arbitrary instance of \setc. We construct an instance of \setck where the universe is $[n]$, the family of sets is \FF, the weight and value of every set in \FF is $1$, and the bag size is $k$. Thus, every feasible solution of the \setck instance can contain at most $k$ sets from \FF. Hence, any $\rho$ approximation algorithm must produce a feasible solution of the \setc instance whenever it exists, thereby solving unweighted \setc in polynomial time, assuming \AA runs in polynomial time.

Similarly, for other problems in the theorem, we can show that one can have a polynomial-time algorithm for the ``base problem" if there exists a $\rho$-approximation algorithm for maximizing the value of the solution of the knapsack version.
\end{proof}

The inapproximability barriers of \Cref{approx-lb,hardnessofapprox} can be overcome using the framework of \FPT-approximation. In particular, we will show \FPT $(1-\eps)$-approximation algorithms, parameterized by the treewidth of the input graph, for all four variants of vertex cover knapsack for maximizing the value of the solution.

%% file: results-pc.tex
\section{Results: Parameterized Complexity}
\label{parameterized}

We study the four variants of \vcknapsack using the framework of parameterized complexity. For that, we consider the treewidth of the input graph as a parameter. With respect to treewidth, we design algorithms that run in time single exponential in the treewidth times polynomial in $n$ (number of vertices), size $s$, and target value $p$ of the knapsack. We then use these algorithms to develop a $(1-\eps)$-approximation algorithm for maximizing the value of the solution that runs in time single exponential in the treewidth times polynomial in the number $n$ of vertices, $1/\eps$ and \(\sum_{v\in\VV}\alpha(v)\).

We know that there exists a $\OO\left(2^{\tw}\cdot \tw^{\OO(1)}\cdot n\right)$ time algorithm for the \vc problem \cite{cygan2015parameterized}. It turns out that it is relatively easy to modify that algorithm to design a similar algorithm \vcknapsack, \vcknapsackbudget, and \minimumvcknapsack.

\begin{theorem}\label{vckparameterzied}
  There is an algorithm for \vcknapsack with running time $\OO\left(2^{\tw}\cdot n^{\OO(1)} \cdot {\sf min}\{s^2,p^2\}\right)$.   
\end{theorem} 
\begin{proof}
Let $(G = (V,E),{(w(u))_{u \in V_G}, (\alpha(u))_{u\in V_G}}, s,p)$ be an input instance of \vcknapsack such that $\tw=tw(G)$. For technical purposes, we guess a vertex $v \in \mathcal{U}$ --- once the guess is fixed, we are only interested in finding solution subsets $\mathcal{U}'$ that contain $v$ and $\mathcal{U}$ is one such candidate. We also consider a nice edge tree decomposition $(\mathbb{T} = (V_{\mathbb{T}},E_{\mathbb{T}}),\mathcal{X})$ of $G$ that is rooted at a node $r$, and where $v$ has been added to all bags of the decomposition. Therefore, $X_{r} = \{v\}$ and each leaf bag is the singleton set $\{v\}$.

We define a function $\ell: V_{\mathbb{T}} \rightarrow \mathbb{N}$ as follows.
For a vertex $t \in V_\mathbb{T}$, $\ell(t) = \sf{dist}_\mathbb{T}(t,r)$, where $r$ is the root. Note that this implies that $\ell(r) = 0$. Let us assume that the values that $\ell$ take over the nodes of $\mathbb{T}$ are between $0$ and $L$. For a node $t\in V_\mathbb{T}$, we denote the set of vertices in the bags in the subtree rooted at $t$ by $V_t$ and $G_t=G[V_t]$. Now, we describe a dynamic programming algorithm over $(\mathbb{T},\mathcal{X})$. We have the following states.

    \textbf{States:} We maintain a DP table $D$ where a state has the following components:
    \begin{enumerate}
    \item $t$ represents a node in $V_\mathbb{T}$.
    \item $\mathbb{S}$ represents a subset of the vertex subset $X_t$
    \end{enumerate}
    \textbf{Interpretation of States}: For each node $t \in \mathbb{T}$, we keep a list $D[t,\mathbb{S}]$ for each $\mathbb{S} \subseteq X_t$, holding the (weight, value) pair of all undominated vertex cover knapsacks $\mathbb{S}'$ of $G_v$ with $\mathbb{S}' \cap X_t = \mathbb{S}$ if such a $\mathbb{S}'$ exist, and $D[t,\mathbb{S}] = \infty$ otherwise. We say one vertex cover $S_1\subseteq V$ dominates another vertex cover $S_2\subseteq V$ if $w(S_1)\le w(S_2)$ and $\alpha(S_1)\ge\alpha(S_2)$ with at least one inequality being strict.

For each state $[t,\mathbb{S}]$, we initialize $D[t,\mathbb{S}]$ to the list $\{(0,0)\}$.

    \textbf{Dynamic Programming on $D$}: We update the table $D$ as follows. We initialize the table for states with nodes $t\in V_\mathbb{T}$ such that $\ell(t)=L$. When all such states are initialized, then we move to update states where the node $t$ has $\ell(t) = L-1$, and so on, till we finally update states with $r$ as the node --- note that $\ell(r) =0$. For a particular $j$, $0\leq j< L$ and a state $[t,\mathbb{S}]$ such that $\ell(t) = j$, we can assume that $D[t',\mathbb{S}']$ have been computed for all $t'$ such that $\ell(t')>j$ and all subsets $\mathbb{S}'$ of $X_{t'}$. Now we consider several cases by which $D[t,\mathbb{S}]$ is updated based on the nature of $t$ in $\mathbb{T}$:
    
    \begin{enumerate}
    \item Suppose $t$ is a leaf node with $X_{t} = \{v\}$. Then $D[t,\{v\}]$ = $(w(v), \alpha(v))$ and $D[t,\phi]$ stores the pair $(0,0)$.

    \item Suppose $t$ is an introduce node. Then it has only one child $t'$ where $X_{t'} \subset X_{t}$ and there is exactly one vertex $u$ that belongs to $X_{t}$ but not $X_{t'}$. Then for all $\mathbb{S} \subseteq X_t$: If $\mathbb{S}$ is not a vertex cover of $G[X_t]$, we set $D[t,\mathbb{S}] = \infty$. Otherwise, if $u \in \mathbb{S}$, then for every pair $(w,\alpha)$ in $D[t',\mathbb{S} \setminus \{u\}]$ with $w + w(u) \leq s$, we add $(w+w(u),\alpha+\alpha(u))$ to the set in $D[t,\mathbb{S}]$. If $u \notin \mathbb{S}$, then we copy all pairs of $D[t',\mathbb{S}]$ to $D[t,\mathbb{S}]$.
    
    \item Suppose $t$ is a forget vertex node. Then it has only one child $t'$, and there is a vertex $u$ such that $X_t\cup\{u\} = X_{t'}$. Then for all $\mathbb{S} \subseteq X_t$, we copy all feasible undominated pairs stored in $D[t',\mathbb{S}]$ to $D[t,\mathbb{S}]$. If any pair stored in $D[t,\mathbb{S}]$ is dominated by any pair of $D[t',\mathbb{S} \cup \{u\}]$, we copy only the undominated pairs to $D[t,\mathbb{S}]$.
  
    \item Suppose $t$ is a join node. Then it has two children $t_1,t_2$ such that $X_t = X_{t_1} = X_{t_2}$. Then for all $\mathbb{S} \subseteq X_t$, let $(w_{\mathbb{S}}, \alpha_{\mathbb{S}})$ be the total weight and value of the vertices in $\mathbb{S}$. Consider a pair $(w_1,\alpha_1)$ in $D[t_1,\mathbb{S}]$ and a pair $(w_2,\alpha_2)$ in $D[t_2,\mathbb{S}]$. Suppose $w_1 + w_2 - w_{\mathbb{S}} \leq s$. Then we add $(w_1 + w_2 - w_{\mathbb{S}}, \alpha_1+\alpha_2-\alpha_{\mathbb{S}})$ to $D[t,\mathbb{S}]$.
    \end{enumerate}
    
Finally, in the last step of updating $D[t,\mathbb{S}]$, we go through the list saved in $D[t,\mathbb{S}]$ and only keep undominated pairs. The proof of correctness of this algorithm is analogous to the proof of correctness of a similar algorithm for \vc.

\textbf{Running time}: There are $n$ choices for the fixed vertex $v$. Upon fixing $v$ and adding it to each bag of $(\mathbb{T}, \mathcal{X})$ we consider the total possible number of states. For every node $t$, we have $2^{|X_t|}$ choices of ${\mathbb{S}}$. For each state, for each $w$, there can be at most one pair with $w$ as the first coordinate; similarly, for each $\alpha\le p$, there can be at most one pair with $\alpha$ as the second coordinate. Thus, the number of undominated pairs in each $D[t,\mathbb{S}]$ is at most ${\sf min}\{s,p\}$ time. Since the treewidth of the input graph \GG is at most \tw, it is possible to construct a data structure in time $\tw^{\OO(1)} \cdot n$ that allows performing adjacency queries in time $\OO(\tw)$.  For each node $t$, it takes time $\OO\left(2^{\tw} \cdot \tw^{\OO(1)} \cdot {\sf min}\{s^2,p^2\}\right)$ to compute all the values $D[t,\mathbb{S}]$ and remove all undominated pairs. Since we can assume w.l.o.g that the number of nodes of the given tree decompositions is $\OO(\tw \cdot n)$, and there are $n$ choices for the vertex $v$, the running time of the algorithm is $\OO\left(2^{\tw}\cdot n^{\OO(1)} \cdot {\sf min}\{s^2,p^2\}\right)$.
\end{proof}

It turns out that the main idea of the algorithm of \Cref{vckparameterzied} can be modified to obtain algorithms for \minimumvcknapsack and \vcknapsackbudget with similar running times.

\begin{theorem}($\star$)\label{minvckparameterzied}
  There is an algorithm for \minimumvcknapsack with running time $\OO\left(2^{\tw}\cdot n^{\OO(1)} \cdot {\sf min}\{s^2,p^2\}\right)$.   
 \end{theorem}

\begin{theorem}\label{vckbudparameterzied}
There is an algorithm for \vcknapsackbudget with running time $\OO\left(2^{\tw}\cdot n^{\OO(1)} \cdot {\sf min}\{s^2,p^2\}\right)$.     
\end{theorem}
\begin{proof}
Let $(G = (V,E),{(w(u))_{u \in V_G}, (\alpha(u))_{u\in V_G}}, s,p)$ be an input instance of \vcknapsack such that $\tw=tw(G)$. We consider a nice tree decomposition $(\mathbb{T} = (V_{\mathbb{T}},E_{\mathbb{T}}),\mathcal{X})$ of $G$ that is rooted at a node $r$. 

We define a function $\ell: V_{\mathbb{T}} \rightarrow \mathbb{N}$ as follows.
For a vertex $t \in V_\mathbb{T}$, $\ell(t) = \sf{dist}_\mathbb{T}(t,r)$, where $r$ is the root. Note that this implies that $\ell(r) = 0$. Let us assume that the values that $\ell$ take over the nodes of $\mathbb{T}$ are between $0$ and $L$. Now, we describe a dynamic programming algorithm over $(\mathbb{T},\mathcal{X})$. We have the following states.

    \textbf{States:} We maintain a DP table $D$ where a state has the following components:
    \begin{enumerate}
    \item $t$ represents a node in $V_\mathbb{T}$.
    \item $\mathbb{S}$ represents a subset of the vertex subset $X_t$
    \item $i$ denotes the size of the vertex cover 
    \end{enumerate}
    \textbf{Interpretation of States}: For each node $t \in \mathbb{T}$ we keep a list of feasible and undominated pairs in $D[t,\mathbb{S},i]$ for each $\mathbb{S} \subseteq X_t$, such that the size of the vertex cover of $G_t$ is at most $i$, holding the vertex cover knapsack $\mathbb{S}'$ of $G_t$ with $\mathbb{S}' \cap X_t = \mathbb{S}$ if such a $\mathbb{S}'$ exist, and $D[t,\mathbb{S},i] = \infty$ otherwise.

For each state $[t,\mathbb{S},i]$, we initialize $D[t,\mathbb{S},i]$ to the list $\{(0,0)\}$. Our computation shall be such that in the end each $D[t,\mathbb{S},i]$ stores the set of all undominated feasible pairs $(w,\alpha)$ for the state $[t,\mathbb{S},i]$.

    \textbf{Dynamic Programming on $D$}: We describe the following procedure to update the table $D$. We start updating the table for states with nodes $t\in V_\mathbb{T}$ such that $\ell(t)=L$. When all such states are updated, then we move to update states where the node $t$ has $\ell(t) = L-1$, and so on till we finally update states with $r$ as the node --- note that $\ell(r) =0$. For a particular $j$, $0\leq j< L$ and a state $[t,\mathbb{S},i]$ such that $\ell(t) = j$, we can assume that $D[t',\mathbb{S}',i]$ have been evaluated for all $t'$ and $i$ such that $\ell(t')>j$ and all subsets $\mathbb{S}'$ of size $i$ of $X_{t'}$. Now we consider several cases by which $D[t,\mathbb{S},i]$ is updated based on the nature of $t$ in $\mathbb{T}$:
    
    \begin{enumerate}
     \item Suppose $t$ is a leaf node with $X_{t} = \{v\}$ . Then $D[t,\{v\},1]$ = $(w(v), \alpha(v))$ and $D[t,\phi,1]$ stores the pair $(0,0)$.

    \item Suppose $t$ is an introduce node. Then it has an only child $t'$ where $X_{t'} \subset X_{t}$ and there is exactly one vertex $v$ that belongs to $X_{t}$ but not $X_{t'}$. Then for all $\mathbb{S} \subseteq X_t$: First, suppose $\mathbb{S}$ is not a vertex cover of $G[X_t]$, we set $D[t,\mathbb{S},i] = \infty$. \\
    Next, suppose $u \in \mathbb{S}$. Then for each pair $(w,\alpha)$ in $D[t',\mathbb{S},i]$, if $w + w(u) \leq s$ we add $(w+w(u),\alpha+\alpha(u))$ to the set in $D[t,\mathbb{S} \setminus \{u\},i-1]$.\\
    Otherwise we copy all pairs of $D[t',\mathbb{S},i]$ to $D[t,\mathbb{S},i]$.
    
     \item Suppose $t$ is a forget vertex node. Then it has an only child $t'$ where $X_t \subset X_{t'}$ and there is exactly one vertex $u$ that belongs to $X_{t'}$ but not to $X_t$. Then for all $\mathbb{S} \subseteq X_t$, we copy all feasible undominated pairs stored in $D[t',\mathbb{S},i]$ to $D[t,\mathbb{S},i]$. If any pair stored in $D[t,\mathbb{S},i]$ is dominated by any pair of $D[t',\mathbb{S} \cup \{u\}, i + 1]$, we copy only the undominated pairs to $D[t,\mathbb{S},i]$. 
  
    \item Suppose $t$ is a join node. Then it has two children $t_1,t_2$ such that $X_t = X_{t_1} = X_{t_2}$. Then for all $\mathbb{S} \subseteq X_t$, let $(w_{\mathbb{S}}, \alpha_{\mathbb{S}})$ be the total weight and value of the vertices in $\mathbb{S}$. Consider a pair $(w_1,\alpha_1)$ in $D[t_1,\mathbb{S},i]$ and a pair $(w_2,\alpha_2)$ in $D[t_2,\mathbb{S},i]$. Suppose $w_1 + w_2 - w_{\mathbb{S}} \leq s$. Then we add $(w_1 + w_2 - w_{\mathbb{S}}, \alpha_1+\alpha_2-\alpha_{\mathbb{S}})$ to $D[t,\mathbb{S},i]$.
    \end{enumerate}
    
Finally, in the last step of updating $D[t,\mathbb{S},i]$, we go through the list saved in $D[t,\mathbb{S},i]$ and only keep undominated pairs.

\textbf{Running time}: There are $n$ choices for the fixed vertex $v$. Upon fixing $v$ and adding it to each bag of $(\mathbb{T}, \mathcal{X})$ we consider the total possible number of states. For every node $t$, we have $2^{|X_t|}$ choices of ${\mathbb{S}}$. For each choice of $\mathbb{S}$, we keep track of the size of vertex cover through invariant $i$ that ranges from 1 to $n$. For each state, for each $w$, there can be at most one pair with $w$ as the first coordinate; similarly, for each $\alpha\le p$, there can be at most one pair with $\alpha$ as the second coordinate. Thus, the number of undominated pairs in each $D[t,\mathbb{S},i]$ is at most ${\sf min}\{s,p\}$ time. Since the treewidth of the input graph \GG is at most \tw, it is possible to construct a data structure in time $\tw^{\OO(1)} \cdot n$ that allows performing adjacency queries in time $\OO(\tw)$.  For each node $t$, it takes time $2^{\tw+1} \cdot \tw^{\OO(1)} \cdot n \cdot {\sf min}\{s^2,p^2\}$ to compute all the values $D[t,\mathbb{S},i]$ and remove all undominated pairs. Since we can assume w.l.o.g that the number of nodes of the given tree decompositions is $\OO(\tw \cdot n)$, the running time of the algorithm is $\OO\left(2^{\tw}\cdot n^{\OO(1)} \cdot {\sf min}\{s^2,p^2\}\right)$.
\end{proof}

However, the approach of \Cref{vckparameterzied} breaks down for \minimalvcknapsack. This is so because a minimal vertex cover (unlike a vertex cover, a vertex cover of size at most $k$, and a minimum vertex cover) of a graph may not be a minimal vertex cover of some of its induced subgraphs. For this reason, it is not enough to keep track of all minimal vertex covers of the subgraphs rooted at some tree node intersecting the bag at certain subsets. Intuitively speaking, we tackle this problem by adding another subset of vertices in the ``indices'' of the DP table that will be part of some minimal vertex cover of some other induced subgraph.

\begin{theorem}\label{fptminimal}
    There is an algorithm for \minimalvcknapsack with running time $\OO\left(16^{\tw}\cdot n^{\OO(1)} \cdot {\sf min}\{s^2,p^2\}\right)$.
\end{theorem}
\begin{proof}
Let $(G = (V,E),{(w(u))_{u \in V}, (\alpha(u))_{u\in V}}, s,p)$ be an input instance of \minimalvcknapsack and $(\mathbb{T} = (V_{\mathbb{T}},E_{\mathbb{T}}),\mathcal{X})$ a nice tree decomposition rotted at node $r$ of treewidth $\tw(G)$.

We define a function $\ell: V_{\mathbb{T}} \rightarrow \mathbb{N}$.
For a vertex $t \in V_\mathbb{T}$, $\ell(t) = \sf{dist}_\mathbb{T}(t,r)$, where $r$ is the root. Note that this implies that $\ell(r) = 0$. Let us assume that the values that $\ell$ take over the nodes of $\mathbb{T}$ are between $0$ and $L$. For a node $t\in V_\mathbb{T}$, we denote the set of vertices in the bags in the subtree rooted at $t$ by $V_t$ and $G_t=G[V_t]$. We now describe a dynamic programming algorithm over $(\mathbb{T},\mathcal{X})$ for \minimalvcknapsack.

    \textbf{States:} We maintain a DP table $D$ where a state has the following components:
    \begin{enumerate}
    \item $t$ represents a node in $V_\mathbb{T}$.
    \item $V_1$, $V_2$ are subsets of the bag $X_t$, not necessarily disjoint.
    \item $V_1$ represents the intersection of $X_t$ with a minimal vertex cover of the subgraph $G_t[(V_t\setminus V_2)\cup V_1]$.
    \end{enumerate}
    \textbf{Interpretation of States:} For each node $t \in \mathbb{T}, V_1, V_2\subseteq X_t$ and ``undominated'' minimal vertex cover $S$ of the induced graph $G_t[(V_t\setminus V_2)\cup V_1]$ such that $S\cap X_t=V_1$, we store $(w(S),\alpha(S))$ in the list $D[t,V_1,V_2]$. We say a minimal vertex cover $S_1\subseteq (V_t\setminus V_2)\cup V_1$ dominates another minimal vertex cover $S_2\subseteq (V_t\setminus V_2)\cup V_1$ if $w(S_1)\le w(S_2)$ and $\alpha(S_1)\ge\alpha(S_2)$ with at least one inequality being strict. We say a minimal vertex cover of $G_t[(V_t\setminus V_2)\cup V_1]$ {\em undominated} if no other minimal vertex cover of $G_t[(V_t\setminus V_2)\cup V_1]$ dominates it.

For each state $D[t,V_1,V_2]$, we initialize $D[t,V_1,V_2]$ to the list $\{(0,0)\}$.

    \textbf{Dynamic Programming on $D$:} We first update the table for states with nodes $t\in V_\mathbb{T}$ such that $\ell(t)=L$. When all such states are updated, then we update states where the level of node $t$ is $L-1$, and so on, till we finally update states with $r$ as the node --- note that $\ell(r) =0$. For a particular $j$, $0\leq j< L$ and a state $[t,V_1,V_2]$ such that $\ell(t) = j$, we can assume that $D[t,V_1,V_2]$ have been evaluated for all $t'$, such that $\ell(t')>j$ and all subsets $V_1^\pr$ and $V_2^\pr$ of $X_{t'}$. Now we consider several cases by which $D[t,V_1,V_2]$ is updated based on the nature of $t$ in $\mathbb{T}$:
    
    \begin{enumerate}
    \item Suppose $t$ is a leaf node with $X_{t} = \{v\}$ . Then $D[t,v,\emptyset]$ = $(w(v), \alpha(v))$, or $D[t,\emptyset,v]$ = $(0,0)$ and $D[t,\emptyset, \emptyset]$ stores the pair $(0,0)$.

    \item Suppose $t$ is an introduce node. Then it has an only child $t'$ where $X_{t^\pr}\cup\{u\} = X_{t}$. Then for all $\mathbb{S} \subseteq X_t$: If $\mathbb{S}$ is not a vertex cover of $G[X_t]$, we set $D[t,V_1,V_2] = (\infty, \infty)$. \\
    Otherwise, we have three cases:
    \begin{enumerate}
        \item Case 1: If $u \not\in V_1 \cup V_2$, then we copy each pair $(w, \alpha)$ from $D[t',V_1,V_2]$
        \item Case 2: If $u \in V_1\setminus V_2$, then 
        \begin{enumerate}
            \item we check if $ N(u) \setminus V_1 \ne \emptyset$, then
        
            \begin{enumerate}
                \item for each pair $(w, \alpha)$ in $D[t',V_1 \setminus \{u\},V_2]$, if $w + w(u) \leq s$, then we put $(w + w(u), \alpha +\alpha(u))$ in $D[t',V_1 \setminus \{u\},V_2]$ to $D[t,V_1,V_2]$.
                
                \item for each pair $(w, \alpha)$ in $D[t',V_1 \setminus \{u\},V_2]$, if $w + w(u) > s$, we put $(w, \alpha)$ to $D[t,V_1,V_2]$.
            \end{enumerate}

            \item Otherwise we store $(\infty, \infty)$.
        \end{enumerate}
        \item Case 3: If $u \in V_2$, then we set $D[t,V_1,V_2]=D[t',V_1, V_2 \setminus \{u\}]$. 
    \end{enumerate}

    \item Suppose $t$ is a forget vertex node. Then it has an only child $t^\pr$ where $X_t\cup\{u\} = X_{t^\pr}$. We copy all the pairs from $D[t',V_1 \cup \{u\}, V_2]$, $D[t',V_1, V_2 \cup \{u\}]$ and $D[t',V_1 , V_2]$ to $D[t,V_1 , V_2]$ and remove all dominated pairs.
  
    \item Suppose $t$ is a join node. Then it has two children $t_1,t_2$ such that $X_t = X_{t_1} = X_{t_2}$. Let $(w(V_1 \cap V_2), \alpha(V_1 \cap V_2))$ be the total weight and value of the vertices in $V_1 \cap V_2$. Then for all $W_1, W_2 \subseteq V_1 \subseteq X_t$, consider a pair $(w_1,\alpha_1)$ in $D[t_1,W_1, W_2 \cup V_2]$ and a pair $(w_2,\alpha_2)$ in $D[t_2,W_2, W_1 \cup V_2]$. Suppose $w_1 + w_2 - w(V_1 \cap V_2) \leq s$. Then we add ($w_1 + w_2 - w(V_1 \cap V_2),  \alpha_1+\alpha_2-\alpha(V_1 \cap V_2)$) to $D[t,V_1, V_2]$.
    \end{enumerate}
    
Finally, in the last step of updating $D[t,V_1, V_2]$, we go through the list saved in $D[t,V_1, V_2]$ and only keep undominated pairs.

\textbf{Correctness of the algorithm:} Recall that we are looking for a solution $\mathcal{U}$ that contains the fixed vertex $v$ that belongs to all bags of the tree decomposition. In each state we maintain the invariant $V_1, V_2 \subseteq X_t$ such that $V_1 = X_t \cap$ minimal vertex cover knapsack of $G_t \setminus $ edges incident on $ V_2 \setminus V_1$. First, we show that a pair $(w,\alpha)$ belonging to $D[t,V_1, V_2]$ for a node $t \in V_\mathbb{T}$ and a subset $\mathbb{S}$ of $X_t$ corresponds to a minimal vertex cover knapsack $H$ in $G_t$. Recall that $X_r = \{v\}$. Thus, this implies that a pair $(w,\alpha)$ belonging to $D[r, V_1 = \{v\}, V_2 = \emptyset )]$ corresponds to a minimal vertex cover knapsack of $G$. Moreover, the output is a pair that is feasible and with the highest value. 

In order to show that a pair $(w,\alpha)$ belonging to $D[t, V_1, V_2]$ for a node $t \in V_\mathbb{T}$ and a subset $V_1$ of $X_t$ corresponds to a minimal vertex cover knapsack of $G_t \setminus $ edges incident on $ V_2 \setminus V_1$, we need to consider the cases of what $t$ can be:

\begin{enumerate}

    \item \textbf{Leaf node}: Recall that in our modified nice tree decomposition we have added a vertex $v$ to all the bags. Suppose a leaf node $t$ contains a single vertex $v$, $D[t,v,\emptyset] = (w(v),\alpha(v))$, $D[t,\emptyset,v] = (0,0)$ and $D[L,\emptyset, \emptyset]$ stores the pair $(0, 0)$. This is true in particular when $j = L$, the base case. From now we can assume that for a node $t$ with $\ell(t) = j < L$ and all subsets $V_1, V_2 \subseteq X_t$, $D[t',V_1^{\pr\pr},V_2^{\pr\pr}]$ entries are correct and correspond to minimal vertex cover in $G_{t'} \setminus $ edges incident on $ V_2^{\pr\pr} \setminus V_1^{\pr\pr}$. when $\ell(t') > j$.

    \item \textbf{Introduce node}: When $t$ is an introduce node, there is a child $t'$. We are introducing a vertex $u$ and the edges associated with it in $G_t$. Since $\ell(t') > \ell(t)$, by induction hypothesis all entries in $D[t',V_1^{\pr\pr} = V_1 \setminus \{u\}, V_2^{\pr\pr} =V_2\setminus \{u\}]$, $D[t', V_1^{\pr\pr}=V_1\setminus \{u\},V_2^{\pr\pr} = V_2]$, and $D[t', V_1^{\pr\pr}= V_1,V_2^{\pr\pr} = V_2 \setminus \{u\}]$, $\forall$ $V_1^{\pr\pr} , V_2^{\pr\pr} \subseteq X_{t'}$ are already computed and feasible. We update pairs in $D[t,V_1,V_2]$ from $D[t',V_1 \setminus \{u\}, V_2]$ or $D[t',V_1,V_2 \setminus \{u\}]$ or $D[t',V_1 \setminus \{u\}, V_2 \setminus \{u\}]$  such that either $u$ is considered as part of a minimal solution in $G_{t} \setminus $ edges incident on $ V_2 \setminus V_1$ or not. 

    \item \textbf{Forget node}:  When $t$ is a forget node, there is a child $t'$. We are forgetting a vertex $u$ and the edges associated with it in $G_t$. Since $\ell(t') > \ell(t)$, by induction hypothesis all entries in $D[t',V_1^{\pr\pr} = V_1 \cup \{u\}, V_2^{\pr\pr} = V_2]$, $D[t',V_1^{\pr\pr}=V_1,V_2^{\pr\pr}=V_2 \cup \{u\}]$, and  $D[t',V_1^{\pr\pr}=V_1,V_2^{\pr\pr}=V_2]$, $\forall$ $V_1^{\pr\pr}, V_2^{\pr\pr} \subseteq X_{t'}$ are already computed and feasible. We copy each undominated $(w,\alpha)$ pair stored in $D[t',V_1 \cup \{u\}, V_2]$, $D[t',V_1, V_2 \cup \{u\}]$ and $D[t',V_1 , V_2]$ to $D[t,V_1,V_2]$.

    \item \textbf{Join node}: When $t$ is a join node, there are two children $t_1$ and $t_2$ of $t$, such that $X_t = X_{t_1} = X_{t_2}$. For all subsets  $V_1 \subseteq X_t$ we partition $V_1$ into two subsets $W_1$ and $W_2$ (not necessarily disjoint) such that $W_1$ is the intersection of $X_{t_1}$ with minimal solution in the graph $G_{t_1} \setminus $ edges incident on $(W_2 \cup V_2) \setminus W_1$. Similarly, $W_2$ is the intersection of $X_{t_2}$ with minimal solution in the graph $G_{t_2} \setminus $ edges incident on $(W_1 \cup V_2) \setminus W_2$. By the induction hypothesis, the computed entries in $D[t_1,W_1, W_2 \cup V_2]$ and $D[t_2,W_2, W_1 \cup V_2]$ where $W_1 \cup W_2 = V_1$ are correct and store the non redundant minimal vertex cover for the subgraph $G_{t_1}$ in $W_1$ and similarly, $W_2$ for $G_{t_2}$. Now we add ($w_1 + w_2 - w(V_1 \cap V_2),  \alpha_1+\alpha_2-\alpha\{V_1 \cap V_2))$) to $D[t,V_1, V_2]$.

\end{enumerate}

What remains to be shown is that an undominated feasible solution $\mathcal{U}$ of \minimalvcknapsack in $G$ is contained in $D[r,\{v\},\emptyset]$. Let $w$ be the weight of $\mathcal{U}$ and $\alpha$ be the value. Recall that $v \in \mathcal{U}$. For each $t$, we consider the subgraph $G_t$ and observe how the minimal solution $\mathcal{S}'$ interacts with $G_t$. Let $\hat{V}_1, \hat{V}_2, \ldots, \hat{V}_m$ be components of $G_t \cap \mathcal{U}$ and let for each $1 \leq i \leq m$, $\mathcal{S}_i = X_t \cap \hat{V}_i$. Also, let $\hat{V}_0 = X_t \setminus \mathcal{U}$. Consider $\mathcal{S} = (\hat{V}_0, \hat{V}_1, \hat{V}_2, \ldots, \hat{V}_m)$. For each $\mathcal{S}_i$, we define subsets $V_1$ and $V_2$ such that $V_1, V_2 \subseteq \mathcal{S}_i$, and $V_1 \cup V_2 = \mathcal{S}_i$, $\forall i \in [m]$. The algorithm updates in $D[t,V_1,V_2]$ the pair $(w',\alpha')$ for the subsolution $(G_{t} \setminus $ edges incident on $ V_2) \cap \mathcal{U}$. Therefore, $D[r,\{v\},\emptyset]$ contains the pair $(w,\alpha)$. Thus, we are done.

\textbf{Running Time:}
There are $n$ choices for the fixed vertex $v$. Upon fixing $v$ and adding it to each bag of $(\mathbb{T}, \mathcal{X})$, we consider the total possible number of states. Observe that the number of subproblems is small: for every node $t$, we have only $2^{|X_t|}$ choices for $V_1$ and $V_2$. Hence, the number of entries of the DP table is $\OO\left(4^\tw \cdot n\right)$. For each state, since we are keeping only undominated pairs, for each weight $w$ there can be at most one pair with $w$ as the first coordinate; similarly, for each value $\alpha$ there can be at most one pair with $\alpha$ as the second coordinate. Thus, the number of undominated pairs in each $D[t,V_1, V_2]$ is at most ${\sf min}\{s,p\}$ that can be maintained in time ${\sf min}\{s^2,p^2\}$. Updating the entries of the join nodes has the highest time complexity among all tree nodes, which is $\OO\left(4^\tw\cdot n^{\OO(1)}\right)$. Hence, the overall running time of our algorithm is $\OO\left(16^\tw\cdot n^{\OO(1)}\cdot {\sf min}\{s^2,p^2\}\right)$. 
\end{proof}

We now design a fully \FPT-time approximation scheme for the \minimalvcknapsack problem by rounding the values of the items so that $\alpha(\VV)$ is indeed a polynomial in $n$. The idea is to scale down the value of every vertex of the input instance so that the sum of values of the vertices that can be in the solution is polynomially bounded by input length and solve the scaled-down instance using the algorithm in \Cref{fptminimal}. We usually scale down the values by dividing by $(\eps \alpha_{\text{max}})/n$. However, this approach does not work for our problems since $\alpha_{\text{max}}$ is a lower bound on the optimal value for classical knapsack but not necessarily for our vertex cover knapsack variants. We tackle this issue by iteratively guessing upper and lower bounds of \OPT thereby incurring an extra factor of poly$\left(\sum_{v\in\VV}\alpha(v)\right)$.

\begin{theorem}\label{thm-fptas}
For every $\eps>0$, there is an $(1-\eps)$-factor approximation algorithm for \minimalvcknapsack for optimizing the value of the solution and running in time $\OO\left(16^{\tw}\cdot \text{poly}\left(n,1/\eps,\log\left(\sum_{v\in\VV}\alpha(v)\right)\right)\right)$ where \tw is the treewidth of the input graph.
\end{theorem}
\begin{proof}
Let $\II= (\GG,(w(u))_{u\in\VV},(\alpha(u))_{u\in\VV},s)$ be an arbitrary instance of \minimalvcknapsack where the goal is to output a minimal vertex cover knapsack \UU of maximum $\alpha(\UU)$ subject to the constraint that $w(\UU)\le s$. Without loss of generality, we can assume that $w(u)\le s$ for every $u\in\VV$. If not, then, letting \XX the set of vertices whose sizes are more than $s$, we remove \XX, take the neighborhood of \XX in the solution, and update bag size appropriately. If the sum of the sizes of the neighboring vertices of \XX is more than $s$, we output that there is no minimal vertex cover that fits a knapsack of size $s$. Let us define $\Phi=\sum_{v\in\VV}\alpha(v)$.

Let $\OPT(\II)$ be the maximum value of any minimal vertex cover whose total size is at most $s$. Our algorithm assumes knowledge of an integer $\lambda$ such that $\OPT(\II)\in[\lambda,2\lambda]$. Without loss of generality, we can assume that $\alpha(u)\le 2\lambda$ for every $u\in\VV$. If not, then, letting \XX the set of vertices whose values are more than $2\lambda$, we remove \XX, take the neighborhood of \XX in the solution, and update bag size appropriately. If the sum of the sizes of the neighboring vertices of \XX is more than $s$, we output that there is no minimal vertex cover that fits a knapsack of size $s$.

We construct another instance $\II^\pr=\left(\GG,(w(u))_{u\in\VV},\left(\alpha^\pr(u)=\left\lfloor \frac{n\alpha(u)}{\eps\lambda}\right\rfloor\right)_{u\in\VV},s\right)$ of \minimalvcknapsack. We compute the optimal solution $\WW^\pr\subseteq\VV$ of $\II^\pr$ using the algorithm in \Cref{fptminimal} and output $\WW^\pr$. Let $\WW\subseteq\VV$ be an optimal solution of $\II$. Clearly, $\WW^\pr$ is a valid (may not be optimal) solution of \II also, since $\alpha(\WW^\pr)\ge p$ by the correctness of the algorithm in \Cref{fptminimal}. We now prove the approximation factor of our algorithm.
\begin{align*}
     \sum_{u\in\WW^\pr}\alpha(u) &\ge \frac{\eps\lambda}{n}\mathlarger{\mathlarger{\mathlarger{\sum}}}_{u\in\WW^\pr} \left\lfloor \frac{n\alpha(u)}{\eps\lambda}\right\rfloor\\
    &\ge \frac{\eps\lambda}{n}\mathlarger{\mathlarger{\mathlarger{\sum}}}_{u\in\WW} \left\lfloor \frac{n\alpha(u)}{\eps\lambda}\right\rfloor&\text{[since $\WW^\pr$ is an optimal solution of $\II^\pr$]}\\
    &\ge \frac{\eps\lambda}{n}\mathlarger{\mathlarger{\mathlarger{\sum}}}_{u\in\WW} \left( \frac{n\alpha(u)}{\eps\lambda}-1\right)\\
    &\ge \left(\mathlarger{\sum}_{u\in\WW}\alpha(u)\right)-\eps\lambda\\
    &\ge \OPT(\II)-\eps\OPT(\II) & \text{[$\lambda\le\OPT(\II)$]}\\
    &=(1-\eps)\OPT(\II)
\end{align*}
Hence, the approximation factor of our algorithm is $(1-\eps)$. We now analyze the running time of our algorithm, assuming we know such a $\lambda$. The value of any optimal solution of $\II^\pr$ is at most
\[\sum_{u\in\VV}\alpha^\pr(u) \le \frac{n}{\eps\lambda}\sum_{u\in\VV}\alpha(u)\le \frac{n}{\eps\lambda}\sum_{u\in\VV}2\lambda=\frac{2n^2}{\eps}. \]

Of course, we do not know any such $\lambda$. We tackle this issue by running the above algorithm for every $\lambda\in\{2^i: i\in\NB\cup\{0\}, 2^i\le\Phi \}$. We know that there exists at least one value of lambda in the above set, which satisfies the assumption of the algorithm that $\OPT(\II)\in[\lambda,2\lambda]$. Hence, at least one run of the algorithm will return a $(1-\eps)$-approximate solution. We check the output of all the runs of our algorithm and output the solution that is indeed a minimal vertex cover, has size at most $s$, and has the maximum values among all valid solutions.
Therefore, the overall running time of our algorithm is $\OO\left(16^{\tw}\cdot \text{poly}(n,1/\eps,\log\left(\sum_{v\in\VV}\alpha(v)\right))\right)$.
\end{proof}

We obtain similar results for the other three variants of vertex cover knapsack for optimizing the value of the solution.

\begin{corollary}\label{vc-ptas}
For every $\eps>0$, there are $(1-\eps)$-factor approximation algorithms for \vcknapsack, \minimumvcknapsack, and \vcknapsackbudget for optimizing the value of the solution and running in time $\OO\left(2^{\tw}\cdot \text{poly}\left(n,1/\eps,\log\left(\sum_{v\in\VV}\alpha(v)\right)\right)\right)$.  
\end{corollary}

\begin{proof}
    For \vcknapsack and \vcknapsackbudget, our algorithms are similar to \Cref{thm-fptas} except, of course, we use the algorithms in respectively \Cref{vckparameterzied,vckbudparameterzied}. For \minimumvcknapsack, we do not know any polynomial-time algorithm to verify whether a given subset of vertices is a minimum vertex cover. Here, we use the algorithm for \vcknapsackbudget for all possible $n$ values of $k$ and output the solution corresponding to the minimum value of $k$ for which the algorithm returns a valid solution.
\end{proof}

It turns out that we can use a similar idea as in \Cref{thm-fptas,vc-ptas} to design an \FPT time $(1+\eps)$-approximation algorithm, parameterized by treewidth, for all the variants of vertex cover knapsack for minimizing the weight of the solution for every $\eps>0$.

\begin{theorem}\label{fptas-weight} For every $\eps>0$, we have the following. 
    \begin{enumerate}
        \item There is a $(1+\eps)$-factor approximation algorithm for \minimalvcknapsack for optimizing the weight of the solution and running in time $\OO\left(16^{\tw}\cdot \text{poly}\left(n,1/\eps,\log\left(\sum_{v\in\VV}w(v)\right)\right)\right)$ where \tw is the treewidth of the input graph.

        \item There are $(1+\eps)$-factor approximation algorithms for \vcknapsack, \minimumvcknapsack, and \vcknapsackbudget for optimizing the weight of the solution and running in time $\OO\left(2^{\tw}\cdot \text{poly}\left(n,1/\eps,\log\left(\sum_{v\in\VV}w(v)\right)\right)\right)$. 
    \end{enumerate}
\end{theorem}

\begin{proof}
    We present the algorithm for \minimalvcknapsack. The algorithms for \vcknapsack, \vcknapsackbudget, and \minimumvcknapsack are similar.

    Let $\II= (\GG,(w(u))_{u\in\VV},(\alpha(u))_{u\in\VV},p)$ be an arbitrary instance of \minimalvcknapsack where the goal is to output a minimal vertex cover knapsack \UU of minimum weight $w(\UU)$ subject to the constraint that $\alpha(\UU)\ge p$.  Let us define $\Phi=\sum_{v\in\VV}w(v)$.

    Let $\OPT(\II)$ be the maximum value of any minimal vertex cover whose total value is at least $p$. Our algorithm assumes knowledge of an integer $\lambda$ such that $\OPT(\II)\in[\lambda,2\lambda]$. Without loss of generality, we can assume that $w(u)\le 2\lambda$ for every $u\in\VV$. If not, then, letting \XX the set of vertices whose sizes are more than $2\lambda$, we remove \XX, take the neighborhood of \XX in the solution, and update bag size appropriately. If the sum of the sizes of the neighboring vertices of \XX is more than $2\lambda$, we output that there is no minimal vertex cover that fits a knapsack of size $2\lambda$. 
    
    We construct another instance $\II^\pr=\left(\GG,\left(w^\pr(u)=\left\lceil \frac{n w(u)}{2\eps\lambda}\right\rceil\right)_{u\in\VV},(\alpha(u))_{u\in\VV},p\right)$ of \minimalvcknapsack. We compute the optimal solution $\WW^\pr\subseteq\VV$ of $\II^\pr$ using the algorithm in \Cref{fptminimal} and output $\WW^\pr$. Let $\WW\subseteq\VV$ be an optimal solution of $\II$. Clearly, $\WW^\pr$ is a valid (may not be optimal) solution of \II also, since $w(\WW^\pr)\le s$ by the correctness of the algorithm in \Cref{fptminimal}. We now prove the approximation factor of our algorithm.

    \begin{align*}
         \sum_{u\in\WW^\pr}w(u) &\le \frac{2\eps\lambda}{n}\mathlarger{\mathlarger{\mathlarger{\sum}}}_{u\in\WW^\pr} \left\lceil \frac{n w(u)}{2\eps\lambda}\right\rceil\\
        &\le \frac{2\eps\lambda}{n}\mathlarger{\mathlarger{\mathlarger{\sum}}}_{u\in\WW} \left\lceil \frac{n w(u)}{2\eps\lambda}\right\rceil&\text{[since $\WW^\pr$ is an optimal solution of $\II^\pr$]}\\
        &\le \frac{2\eps\lambda}{n}\mathlarger{\mathlarger{\mathlarger{\sum}}}_{u\in\WW} \left( \frac{n w(u)}{2\eps\lambda}+1\right)\\
        &\le \left(\sum_{u\in\WW}w(u)\right)+2\eps\lambda\\
        &\le \OPT(\II)+\eps\OPT(\II) & \text{[$2\lambda\le\OPT(\II)$]}\\
        &=(1+\eps)\OPT(\II)
    \end{align*}

    Hence, the approximation factor of our algorithm is $(1+\eps)$. We now analyze the running time of our algorithm, assuming we know such a $\lambda$. The value of any optimal solution of $\II^\pr$ is at most
    \[\sum_{u\in\VV}w^\pr(u) \le \frac{n}{2\eps\lambda}\sum_{u\in\VV}w(u)\le \frac{n}{2\eps\lambda}\sum_{u\in\VV}2\lambda=\frac{n^2}{\eps}. \]

    Again, as in \Cref{thm-fptas}, running the above algorithm for every $\lambda\in\{2^i: i\in\NB\cup\{0\}, 2^i\le\Phi \}$ and outputting the best valid solution gives a $(1+\eps)$-approximation algorithm with runtime $\OO\left(16^{\tw}\cdot \text{poly}(n,1/\eps,\log\left(\sum_{v\in\VV}w(v)\right))\right)$.
\end{proof}

%% file: conclusion.tex
\section{Conclusion}
\label{conclusion}
 We have studied the classical Knapsack problem with the graph theoretic constraints, namely vertex cover and its interesting variants like \minimumvcknapsack, \minimalvcknapsack, and \vcknapsackbudget. We further generalize this to hypergraphs and study \setck and \hsk. We have presented approximation algorithms for minimizing the size of the solution and proved that the approximation factors are the best possible that one hopes to achieve in polynomial time under standard complexity-theoretic assumptions. However, to maximize the value of the solution, we obtain strong inapproximability results. Fortunately, we show that there exist \FPT algorithms parameterized by the treewidth of the input graph (for vertex cover variants of knapsack), which can achieve $(1-\eps)$-approximate solution.